\documentclass[a4paper,10pt]{article}

\usepackage[pages=all, color=black, position={current page.south}, placement=bottom, scale=1, opacity=1, vshift=5mm]{background}
\SetBgContents{
} 
\usepackage[margin=0.92in]{geometry}

\usepackage{amsmath}
\usepackage{amsthm}
\usepackage{amssymb}
\usepackage{soul}

\usepackage[utf8]{inputenc}
\usepackage{hyperref}
\hypersetup{
	unicode,
	pdfauthor={Author One, Author Two, Author Three},
	pdftitle={A simple article template},
	pdfsubject={A simple article template},
	pdfkeywords={article, template, simple},
	pdfproducer={LaTeX},
	pdfcreator={pdflatex}
}

\usepackage[sort&compress,numbers,square]{natbib}

\theoremstyle{plain}
\newtheorem{theorem}{Theorem}[section]
\newtheorem{corollary}[theorem]{Corollary}
\newtheorem{lemma}[theorem]{Lemma}

\newtheorem{proposition}[theorem]{Proposition}

\theoremstyle{definition}
\newtheorem{definition}[theorem]{Definition}
\newtheorem{example}[theorem]{Example}
\newtheorem{remark}[theorem]{Remark}

\usepackage{graphicx, color}
\graphicspath{{fig/}}

\usepackage{algorithm, algpseudocode} 
\usepackage{mathrsfs}

\definecolor{vert}{rgb}{0.42, 0.557, 0.137}

\title{Lifts of quantum CSS codes}
\author{Virgile Guemard$^{1, 2}$}

\date{
	$^1$Aix Marseille Université, I2M, UMR 7373, 13453 Marseille, France\\%
        $^2$Inria Paris, France \\[2ex]%
	\today
}

\usepackage{dcolumn}% Align table columns on decimal point
\usepackage{bm}% bold math
\usepackage[english]{babel}
\usepackage[OT1]{fontenc}

\usepackage[colorinlistoftodos, color=green!40, prependcaption]{todonotes}

\usepackage{amsthm}
\usepackage{mathtools}
\usepackage{physics}
\usepackage{xcolor}
\usepackage{graphicx}
\usepackage{tikz-cd}
\usepackage{bbold}

\usepackage{adjustbox}
\usepackage{placeins}

\usepackage{csquotes}
\usepackage{booktabs}
\usepackage[caption=false]{subfig}
 \makeatletter
\renewcommand{\fnum@figure}{FIG. \thefigure}
\makeatother
\usepackage{tikz}
\usetikzlibrary{decorations.markings}

\usepackage{amsmath,mleftright}
\usepackage{xparse}

\NewDocumentCommand{\evalat}{sO{\big}mm}{%
  \IfBooleanTF{#1}
   {\mleft. #3 \mright|_{#4}}
   {#3#2|_{#4}}%
}

\usepackage{amsthm}

\usepackage{extpfeil}
\usepackage{tikz-cd}
\usepackage{svg}

\usepackage{titlesec}
\titleformat{\section}%[runin]
  {\normalfont\bfseries}{\thesection}{1em}{}
\titleformat{\subsection}%[runin]
  {\normalfont\bf}{\thesubsection}{1em}{}{\normalfont\bfseries}
\titleformat{\subsubsection}%[runin]
  {\normalfont\bfseries}{\thesubsubsection}{1em}{}
\renewcommand{\thesubsection}{\thesection.\arabic{subsection}}
\renewcommand{\thesection}{\arabic{section}}
\renewcommand{\thesubsubsection}{\thesubsection.\arabic{subsubsection}} 

\usetikzlibrary{arrows}

\usepackage{leftindex}
\makeatletter
\newtheorem*{rep@theorem}{\rep@title}
\newcommand{\newreptheorem}[2]{%
\newenvironment{rep#1}[1]{%
 \def\rep@title{#2 \ref{##1}}%
 \begin{rep@theorem}}%
 {\end{rep@theorem}}}
\makeatother
\newreptheorem{proposition}{Proposition}

\begin{document}
	\maketitle

\begin{abstract}

We propose a notion of lift for quantum CSS codes, inspired by the geometrical construction of Freedman and Hastings  \cite{Freedman2020CSS_Manifold}. It is based on the existence of a canonical complex associated to any CSS code, that we introduce under the name of Tanner cone-complex, and over which we generate covering spaces. As a first application, we describe the classification of lifts of hypergraph product codes (HPC) and demonstrate the equivalence with the lifted product code (LPC) of Panteleev and Kalachev \cite{Panteleev2021}, including when the linear codes, factors of the HPC, are Tanner codes. As a second application, we report several new non-product constructions of quantum CSS codes, and we apply the prescription to generate their lifts which, for certain selected covering maps, are codes with improved relative parameters compared to the initial one.
\end{abstract}
\tableofcontents

\section{Introduction}

\subsection{Context}

Recent years have witnessed growing effort in constructing low-density parity-check (LDPC) quantum error correcting codes, as well as decoding and implementing them fault tolerantly. This has led, in particular, to the findings of LDPC codes with asymptotically good performances, which are families of quantum CSS codes with constant rate and relative distance. Apart from being good candidates for implementation in full scale fault-tolerant quantum computer architectures, they also constitute inspiration to devise new codes of moderate length. \par

The historical development leading to these codes reveals various methods. First, with the work of Hastings, Haah and O'Donnell \cite{Hastings2020}, the CSS codes relied on an algebraic construction inspired by the definition of fiber bundles in geometry. It was then understood that the tensor product of modules over Abelian group rings, introduced by Panteleev and Kalachev  \cite{Panteleev2020,Panteleev2021degeneratequantum} under the name of lifted product, could reach a wide parameter space, which includes short length codes of high rate and distances. Then, the related tensor product of modules over non-commutative rings, or balanced product, was proposed by Breuckmann and Eberhardt  \cite{Breuckmann2021} as a solution to reach potentially better code parameters. Independently, Panteleev and Kalachev \cite{Panteleev2021}, built a family of codes with asymptotically good parameters using the balanced product of two linear codes. Using the operation of lift of linear codes, these two codes are built as modules over a non-commutative group ring. A different strategy was soon after undertaken by Leverrier and Z\'emor \cite{Leverrier2022,leverrier2022decoding}, generalizing the classical Tanner construction \cite{Tanner1981ARA}, and simplifying the previous approach of \cite{Panteleev2021}.\par

In these constructions, the methodology relies on using closely related underlying components, which are geometrical complexes, called left-right square complexes  \cite{Dinur2021}, built from the product of two graphs or an algebraic analogue designated as balanced product complex, combined with local linear codes. For instance, in \cite{Panteleev2021}, these two graphs are generated as the lift of linear codes, endowing them with a chosen group action. This set of constituents, still under studies, appears as a cornerstone of quantum coding theory, with several follow-up constructions and analysis \cite{LinHsie2022,Dinur2023,panteleev2024maximally}.\par 

In classical coding theory, the lift of a code \cite{Thorpe2003, Pusane2011} is a practical and natural method for generating families of LDPC codes with dimension and distance increasing linearly in the length of the code, from a single input code represented by its Tanner graph, often called a protograph in this setting. Geometrically, lifting a linear code is equivalent to taking a covering of its Tanner graph representation. A consequence of this method is that the Tanner graphs of the input and output codes are locally homeomorphic. Asymptotically good families of linear codes built from geometrical lifts exist. For example, the expander codes of Sipser and Spielman \cite{Sipser1996} can be constructed from Cayley graphs, which are geometrical coverings of a simple base graph. This produces, in the best case, a family of expander graphs or even Ramanujan expander graphs \cite{Lubotzky1988}, which are regular graphs with respectively strong and optimal connectivity. For all these reasons, the lift is a central tool for building new codes from old ones in classical coding theory.\par 

Up until now, quantum CSS codes did not enjoy a general operation taking as input any CSS code and outputting another code of increased length, such that the Tanner graphs of the input and output CSS codes are locally homeomorphic. The orthogonality constraint of the two linear codes, required to describe a quantum CSS code, makes a lift of a CSS code much harder to define than for a linear code. Trying to apply the classical recipe on the Tanner graph representation does not always yield a valid quantum CSS code: the orthogonality constraint of the two linear codes can be lost. Only in some particular cases of CSS constructions, a notion of lift is already defined. Firstly, for product codes, such as hypergraph product codes (HPC) of Tillich and Z\'emor \cite{Tillich2009} and their generalization, lifted products codes (LPC) \cite{Panteleev2021}, a lift is an operation that consists in applying the classical lift on each of the linear factor codes. Secondly, for a code defined by the cellulation of a surface, a lift is naturally obtained via a covering of that surface. It is interesting to ask if these two examples can be found under a unified framework.

\subsection{Contributions}\label{Section Contributions and outline}

This article introduces the concept of lift of a quantum CSS code. The key ingredient to define the lift of a quantum CSS code is a geometrical object faithfully representing the code, similarly to the fact that a lift of a linear code can be obtained from a covering of its Tanner graph. Any CSS code admits a geometrical representation provided by a complex that we introduce under the name of \textit{Tanner cone-complex}, in Definition \ref{definition Tanner cone-complex 1}. The lift is then obtained from a covering of that complex. \par
A non-trivial lift of a CSS code only exists if its Tanner cone-complex admits a non-trivial fundamental group. When this is satisfied, the lift of a code enjoys the following properties:
\begin{enumerate}
    \item     For an input CSS code of length $n$, the lift is a code of size given by an integer multiple of $n$,
    \item The maximum weight of rows and columns of the lifted check matrices is unchanged compared to that of the input code,
    \item   Applied to a classical code, it coincides with the lift of linear codes,
    \item Applied to an HPC, it coincides with the LPC construction described in \cite{Panteleev2021}.
\end{enumerate}
Our definition of lift of CSS codes, based on covering maps of the Tanner cone-complex, is therefore a valid generalization of known constructions. Furthermore, we show how the lift can be related, in some cases, to a certain type of fiber bundle code \cite{Hastings2020} and we draw another parallel with the balanced product code \cite{Breuckmann2021}.\par

Originally, the idea of lift of a quantum CSS code came from the study of an article by Freedman and Hastings \cite{Freedman2020CSS_Manifold}, which provides a manifold representation to any CSS code. A natural way to lift a code is then given by generating coverings of this manifold. In a separate note \cite{Guemard2023Unpublished}, we show how this is related to a covering of the Tanner cone-complex.\par

The parameters, dimension and distance, of a lifted CSS code are, in general, hard to determine. One can only hope to do a case by case analysis. As a first step towards understanding the evolution of code parameters, we apply the construction to an arbitrary HPC. We show, in that case, that the lift is more general than simply lifting the two linear codes, factors of the HPC, independently, and we provide a complete classification of lifts of HPCs based on the notion of Goursat quintuples \cite{Johnson1990}, which gives a correspondence between subgroups of group product and quintuples composed of subgroups of the factor group. By analyzing this construction, we find out that the underlying structure is similar to a left-right square complex at the center of the state-of-the-art CSS codes of \cite{Panteleev2021} and good locally testable codes of Dinur et al. \cite{Dinur2021}. This is formulated in Proposition \ref{Proposition Antidiagonal action Goursat}. In particular, Corollary \ref{Corollary linear parameters} shows that the asymptotic parameters of a family of codes generated by the lift of a single input code can, in principle, be linear.\par

Secondly, the lift of a quantum code can be expressed as a left or a right module over a certain group ring, when the associated covering of its Tanner cone-complex is regular. As a result, it is possible to apply the tensor product of modules on two such codes. This constitutes a systematic way of generating balanced product of quantum CSS codes, as first pointed out in \cite{Breuckmann2021}. This is formulated in Definition \ref{Definition balanced product of CSS codes}. The hope of such a construction is to improve on the tensor product operation of quantum codes studied in \cite{Zeng2018,Audoux2015}. In this article, we are not able to claim this result. We therefore leave this analysis for future work.\par

Finally, in Section \ref{section New constructions and lifts}, we introduce three new families of codes which are specifically designed to be lifted. For clarity, we refer to a code built from 2D cellular chain complex as of $\operatorname{E}$-type, since the qubits are assigned to edges. The new type of code that we introduce here are referred to as of $\operatorname{V}$-type, because the qubits, as well as $X$ and $Z$-checks are assigned to the vertices. For each one of these codes, the associated Tanner cone-complex has a fundamental group isomorphic to a predetermined infinite group.\par
The first family, although greatly inspired by topology, cannot be given as the cellular chain complex of any 2D cell complex. The second family is a generalization of the first, having nevertheless a less straightforward connection to topology. The third family is built from similar ideas, but the assignment of qubits and checks to vertices is done according to a different configuration. For each of them, the row and column weight of the parity check matrices can be made strictly greater than 2, so that cells cannot be seen as edges of the 1-skeleton, i.e. a graph, of a 2D complex. The lift of these abstract CSS codes, is hence a relevant and practical way to build a diversity of codes with increased length and parameters, and endowed with a certain group action. For some chosen input codes with low enough weight parity-check matrices, we generate all possible lifts of degree (the degree of the associated covering) 1 to 59 for the first family,  and only lifts of specific degrees for the second and third. This produces codes with moderate length (less than 2000 qubits) and we only report the ones with interesting parameters. The criteria that we choose to compare our codes is the quantity $\frac{kd^2}{n}$, where $n$ is the length of the code, $k$ its dimension and $d$ its distance. An upper bound for the distance is calculated in GAP \cite{GAP4} with the package $\mathtt{QDistRnd}$ \cite{Pryadko2022}. We leave a more rigorous analysis of these new families for future work.\par

\subsection{Outline}

The article is organized as follows. Section \ref{section Preliminaries} reviews generalities on linear and quantum coding theory, where the emphasis is put on their chain complex formulation. In Section \ref{section geometrical lifts of linear code}, we set the methodology for lifting a graph and linear codes with voltage assignments. \par

Section \ref{section Lift of a quantum CSS code} is the core of our paper, where we dive into the new object introduced under the name of Tanner cone-complex, and define the lift of a CSS code. The explicit construction of a lifted code is exposed in Section \ref{section Lift of a CSS code: explicit construction}. Sections \ref{section relation to fiber bundle codes} and \ref{section relation to balanced product codes} give alternative formulations of the lift related to fiber bundle and balanced product codes.\par

Section \ref{section Applications} presents two applications of the lift. The first one is a classification result for lifts of HPCs, addressed in Section \ref{section classification of lifts of HPC}, which also contains a detailed study showing the link between the lift of CSS codes and the LPC constructions. In Section \ref{section balanced product of quantum CSS codes} we present how to obtain the balanced product of two quantum CSS codes.

Lastly, Section \ref{section New constructions and lifts} is where we introduce newly discovered CSS codes and generate lifted codes. This section is divided in three parts. The first presents the general procedure to build the codes and analyze their parameters numerically. The second and third part introduce explicit examples.

\section{Preliminaries}\label{section Preliminaries}

\subsection{Chain complexes}\label{section chain complexes}

In this work, a \textit{chain complex} $C$ is defined as a sequence $(C_{\bullet },\partial_{\bullet })$ of $\mathbb{F}_2$-vector spaces and linear maps, 

\begin{equation}
C:= \cdots \longrightarrow 
C_{i+1} \stackrel{\partial_{n+1}}{\longrightarrow}
C_i \stackrel{\partial_n}{\longrightarrow}
C_{i-1} \stackrel{\partial_{n-1}}{\longrightarrow}
\cdots,
\end{equation} 
such that the composition of any two consecutive maps satisfies $ \partial_i \circ \partial_{i+1}=0$ for all $n$. Objects in $C_i$ are called $i$-\textit{chains} and the linear maps $\partial_i$ are called the \textit{boundary maps}. The index $i$ of a $i$-chain is referred to as the degree. \par

The spaces $C_i$ may have extra structure. For example, they may be modules over a common ring $R$. In that case, the boundary maps must preserve this extra structure: they must be module homomorphisms. We will, for example, encounter situations in which $R$ is a group algebra $\mathbb{F}_2[G]$. \par
Unless otherwise stated, each space $C_i$ has finite dimension $n_i$. It is always given with a preferred basis $\mathcal{B_i}$ and can be identified with $ \mathbb{F}_2^{n_i}$. We can also represent each boundary map $\partial_i$ by a matrix $\operatorname{Mat}_\mathcal{B_i}(\partial_i)$, and for simplicity we identify it with this representation. Moreover, each space is endowed with a symmetric bilinear form, $\langle,\rangle: C_i\times C_i\to \mathbb{F}_2$ given on two $i$-chains by $\langle c,c'\rangle= c_1 c_1'+c_2c_2'+\dots +c_{n_i}c_{n_i}' $, where addition and multiplication are in $\mathbb{F}_2$.\par

For each degree $i$, we define two subspaces of $C_i$:
\begin{itemize}
\item  $Z_i = \operatorname{Ker} \partial_i$, whose elements are called the \textit{cycles},
\item $ B_i =\operatorname{Im}  \partial_{i+1}$, whose elements are called the \textit{boundaries}.
\end{itemize}
From the composition property of consecutive boundary maps, the following embedding relation $ B_i \subseteq Z_i \subseteq C_i$ holds for all $i$.  The $i$-th \textit{homology group} is defined as the vector space, 
\begin{equation}
H_i(C) = Z_i / B_i = \operatorname{Ker} \partial_i / \operatorname{Im}  \partial_{i+1}.
\end{equation}
\par
The \textit{cochain complex} is the dual of a chain complex and defined as a sequence $(C^{\bullet },\delta_{\bullet })$ of dual vector spaces, whose objects are called \textit{cochains}, and linear maps called \textit{coboundary maps}. It is obtained from a chain complex $C$ by replacing each vector space $C_i$ by its dual $ C^{i}:=\operatorname{Hom} (C_{i},\mathbb{F}_2)$, and $\partial_{i}$ by its dual map $\delta_{i}:C^{i-1}\to C^{i}$,
\begin{equation}
C^\bullet= \cdots \longleftarrow 
C^{i+1} \stackrel{\delta_{i+1}}{\longleftarrow }
C^i \stackrel{\delta_i}{\longleftarrow }
C^{i-1} \stackrel{\delta_{i-1}}{\longleftarrow }
\cdots.
\end{equation}
The composition of any two consecutive coboundary maps is the zero map: $ \delta_{i} \circ \delta_{i-1}=0$ .  Finite dimensional vector spaces and their dual spaces are isomorphic, and here an isomorphism can be constructed from the bilinear form $\langle,\rangle$ on $C_i$: this is the linear map given on a $i$-chain by $c\mapsto \langle c, \rangle \in \operatorname{Hom}(C_i,\mathbb{F}_2)$. For each dual space $C^i$, we fix a basis given by $\{\langle c, \rangle \:|\: c\in \mathcal{B_i}\}$ so that we have the identification $C_i=C^i$, and endow it with the same symmetric bilinear form as $C_i$. This way, the matrix representations of the boundary and coboundary maps of same degree are transposed of each other, i.e. $\delta_i=\partial^T_i$.\par

For each degree $i$, we define two subspaces of $C^i$:
\begin{itemize}
\item  $Z^i = \operatorname{Ker} \delta_{i+1}$ whose elements are called the \textit{cocycles},
\item $ B^i = \operatorname{Im}    \delta_{i}$ whose elements are called the \textit{coboundaries}.
\end{itemize}
From the composition  property of consecutive coboundary maps, the embedding relation $ B^i \subseteq Z^i \subseteq C^i$ holds for all $i$ and from commutativity we can define the $i$-th \textit{cohomology group} $H_i(C)$ by 
\begin{equation}
H^i(C) = Z^i / B^i = \operatorname{Ker} \delta_{i}/\operatorname{Im}   \delta_{i-1}.
\end{equation}
\par
When the complex is a sequence of finite dimensional vector spaces, which is the case unless otherwise stated, the $i$-th homology and cohomology groups have the same dimension and are therefore isomorphic as vector spaces, i.e. $H_i(C)\cong H^i(C)$\footnote{ This is not true in the general case.}.\par

Lastly, there exists a product operation for pairs of chain complexes. Let $C$ and $C'$ be two chain complexes. The homological product $C\otimes C'$ is the chain complex $(A_\bullet,\partial^A_\bullet)$ such that
\begin{equation}
A_i=\bigoplus_j C_j\otimes C_{i-j}'.
\end{equation}
It has a boundary operator $\partial_i^A$ acting on $A_i$, such that its restriction on a summand $C_j\otimes C_{i-j}'$ is given by
\begin{equation}\label{equation boundary of tensor product}
\partial^A_{i}|_{j,i-j}=\partial_j\otimes \operatorname{id}_{C'} + (-1)^{i-j} \operatorname{id}_C\otimes\partial_{i-j}'.
\end{equation}
In our case, vector spaces are defined over $\mathbb{F}_2$ and the alternating sign in the second term can be omitted. \par 
The tensor product combines a $j$-cycle of $C$ and an $(i-j)$-cycle of $C'$ to create an $i$-cycle in $C\otimes C'$. The Künneth theorem for field coefficients relates the homology of the product complex to the homology of $C$ and $C'$:
\begin{equation}
 H_{i}(C\otimes C')\cong    \bigoplus _{j}H_{j}(C)\otimes H_{i-j}(C'). 
\end{equation}
Finally, product complexes can also be defined for pairs of chain complexes which are sequences of modules over a group algebra $R$, using the tensor product of modules over that algebra. In that case, the complex is sometimes called the balanced product, and is written $C\otimes_{R }C'$ \cite{Brown1982}. Its boundary maps are defined similarly to Equation \eqref{equation boundary of tensor product}.

\subsection{Linear codes}\label{section linear codes}

In this section, we recall some definitions and set up notations related to linear codes, as they play a key role in the construction of quantum CSS codes. \par
A linear $[n,k]$ code $C$ is a $k$-dimensional subspace of a $\mathbb{F}_2$-space $E$.  In this work, $E$ is always given with a basis $\mathcal{B}$, and is identified with $\mathbb{F}_2^n$. Its parameters $n$ and $k$ are respectively called its length and dimension. The Hamming distance between two vectors in $E$ is equal to the number of coordinates on which they differ and the Hamming weight, or norm, of a vector $c\in E$, noted $|c|$, is equal to its number of non-zero coordinates. The minimal weight of a non-zero code word of $C$ is called the distance of the code, $d(C)=\min\limits_{c\in C\setminus \{0\} }|c|$. The parameters of $C$ are written $[n, k, d]$.\par

The space $E$ is endowed with a symmetric bilinear form $\langle ,\rangle: E\times E\rightarrow \mathbb{F}_2 $, given by $\langle c,c'\rangle= c_1 c_1'+c_2c_2'+\dots +c_nc_n' $, with operations in $\mathbb{F}_2$. Then a linear code can be defined by the image of a linear map, represented by a generator matrix $G$, or the kernel of a linear map, represented by a parity-check matrix $H$, i.e. $C = \{Gy \: | \: y \in \mathbb{F}_2^k\}= \{ c\in \mathbb{F}_2^n \: | \: Hc=0\}$.  The generator and parity-check matrices satisfy $H G^T=0$. For any linear code $C$, the generator matrix of $C$ is the parity-check matrix of its orthogonal complement, called the dual $C^\perp=\{e\in E\: |\: \forall c\in C, \langle e,c\rangle=0\}$. The parameters of the dual code are written $[n^\perp,k^\perp,d^\perp]$. We adopt the convention $d=0$ whenever $k=0$.

Linear codes can be described in the language of chain complex. We adopt here standard notations used in homological algebra, and refer to \cite{Audoux2015, Breuckmann2021review} for a review. We can describe a linear code $C$ by a chain complex, 
\[  C:=C_1\xrightarrow{\partial_1} C_0. \]
The correspondence with the former formulation is given by the following identification: $C_0=\mathbb{F}_2^{m}$ is called the check space, $C_1=\mathbb{F}_2^n $ represents $E$, $\operatorname{Mat}_\mathcal{B}(\partial_1)=H$ and the code space is given by the space of cycles or the homology group,  $\operatorname{ker}(\partial_1)=Z_1(C)=H_1(C)$. The space $C_1$ is also endowed with a symmetric bilinear form $\langle ,\rangle: C_1\times C_1\rightarrow \mathbb{F}_2 $. In this work, we shall use the same symbol, $C$, to either talk about the linear code or its chain complex description, since both can be specified completely by a parity-check matrix. The dual chain complex,
\[C^*=C^1\xleftarrow[]{\partial_1^T}C^0,\]
represents the code with checks and bits interchanged. We can also make the identification $C_i=C^i$, $i=1,2$ due to the isomorphism $x\mapsto \langle x, \rangle \in \operatorname{Hom}(C_i,\mathbb{F}_2) $ and the choice of basis on $C^i$ induced by this map. The code $C^*$ is called the transposed code, following \cite{Tillich2009}. We can also write another chain complex using the generating map of a code, namely $D=D_2\xrightarrow{\partial_2} D_1$, where $\operatorname{Mat}_\mathcal{B}(\partial_2)= G$ and $D_1=\mathbb{F}_2^n $. In that case, the code space corresponds to the space of boundaries $B_1=\operatorname{Im}(\partial_2)$, and $H_1(D)=\operatorname{coker}(\partial_2)$.

\subsection{Quantum CSS codes}\label{section quantum CSS codes}

CSS codes are instances of stabilizer quantum error correcting codes. They first appeared in the seminal work of Calderbank, Shor and Steane \cite{Calderbank1996,Steane1996,Stean1996Multiple}. These quantum codes are defined by two linear codes $C_X$ and $C_Z$ respecting the orthogonality condition $C_X^\perp\subseteq C_Z$. It was later understood that this condition can be formulated in the language of chain complexes \cite{Kitaev2003,Freedman2001}.

\begin{definition}[CSS code]\label{Def_CSS stabilizer code}
Let $C_X$ and $C_Z$ be two linear codes with parity-check matrices $H_X$ and $H_Z$, such that $C_X^\perp\subseteq C_Z \subset \mathbb{F}_2^{n}$ (or equivalently $C_Z^\perp\subseteq C_X$). The CSS code $\text{CSS}(C_X,C_Z)$ is the subspace of $(\mathbb{C}_2)^{\otimes n}$ given by\[\operatorname{Span}\left \{\sum_{u\in C_Z^\perp }\ket{c+u}\text{ }|\text{ } c\in C_X \right \}.\]
Its parameters, noted $[[n,k,d]]$, are:
\begin{itemize}
\item the \textit{code length} $n$,
\item the \textit{dimension} $k = \dim(C_X / C_Z^\perp)=n-\operatorname{rank} H_X-\operatorname{rank} H_Z$,
\item the \textit{distance} $d = \min(d_X, d_Z)$, with
\begin{align*}
d_X=&\min\limits_{c\in C_Z\setminus C_X^\perp}|c|,\\
d_Z=&\min\limits_{c\in C_X\setminus C_Z^\perp}|c|.
\end{align*}
\end{itemize}
A family of CSS codes is said to be $good$ when its parameters are asymptotically $[[ n, \Theta(n), \Theta(n)]]$. The maximum row weight and column weight in the parity check matrix $H_X$, respectively $H_Z$, are noted $w_X,q_X$, respectively $w_Z, q_Z$. A family is called \textit{Low Density Parity Check} (LDPC) if $w_X,q_X,w_Z$ and $q_Z$ are upper bounded by a constant.
\end{definition}
When we know both the $X$ and $Z$ distances, we note the parameters $[[n,k,(d_X,d_Z)]]$. Our convention is to set $d=0$ whenever $k=0$\footnote{This is to avoid ambiguity when we later compute the quantity $kd^2/n$.}. A potential objective, when designing quantum CSS codes, is to obtain the largest possible dimension and distance for a given number of physical qubits $n$. Indeed, the dimension corresponds to the number of logical qubits while the distance is related to the number $t$ of correctable errors by $t=\lfloor (d-1)/2 \rfloor$.\par

Physically, the rows of $H_X$ induce $X$-type operators called parity-checks, or $X$-checks, and the rows of $H_Z$ induce $Z$-checks. The orthogonality property of $C_X$ and $C_Z$ , equivalent to $H_X H_Z^T=0$, is the necessary property for the syndrome of the two codes to be obtained independently, i.e. by commuting operators.\par

We denote a preferred basis for the $Z$-checks, the qubits and the $X$-checks as $Z$, $Q$ and $X$, where the support of an element in $Z$ or $X$ corresponds respectively to a row of $H_Z$ or $H_X$. The orthogonality condition, $H_X H_Z^T=0$, is analogous to the composition property of two boundary maps in a chain complex. Because the data of two parity-check matrices is sufficient to construct a CSS code, such a code is therefore naturally defined as a chain complex or its dual complex,
\begin{equation}
\begin{split}
& \mathbb{F}_2^{m_Z}\xrightarrow{\partial_{2}=H_Z^T}\mathbb{F}_2^{n}\xrightarrow{\partial_{1}=H_X} \mathbb{F}_2^{m_X} , \\
& \mathbb{F}_2^{m_Z}\xleftarrow{\delta_{2}=H_Z}\mathbb{F}_2^{n}\xleftarrow{\delta_{1}=H_X^T} \mathbb{F}_2^{m_X},
\end{split}
\end{equation}
where $n=|Q|$, while $m_X$ and $m_Z$ are the number of $X$ and $Z$-checks. It will also be appropriate to define the chain complex of the code in terms of abstract cells taken directly from the sets of checks and qubits, \[C=\mathbb{F}_2Z\xrightarrow[]{\partial_2} \mathbb{F}_2Q\xrightarrow[]{\partial_1} \mathbb{F}_2X,\]
where $\mathbb{F}_2S:=\bigoplus_{s\in S}\mathbb{F}_2 s $ denotes formal linear combination, called \textit{chains}, of abstract basis cells in the sets $S=X,Q$ or $Z$ and the boundary map is defined by $\partial_i s=\sum_{t\in\operatorname{supp}(s)} t$. In this context, a check and a qubit are identified with the abstract cells representing them. Consequently, for a check $s$, $\operatorname{supp}(s)$ refers to the support of the row vector representing the check in the corresponding parity-check matrix, which can be identified with $\operatorname{supp(\partial_i s})$, the support of the chain $\partial_i s$. \par
Reciprocally, we can extract quantum CSS codes from based chain complexes with coefficient in $\mathbb{F}_2$. Such a chain complex $C$ is said to be a \textit{3-term complex} as it contains three vector spaces (or modules) related by two boundary maps, hence of the form $C=
C_{i+1}\xrightarrow{\partial_{i+1}}C_i\xrightarrow{\partial_i}C_{i-1}$. Note that any chain complex $C $ of length greater than three can be truncated into a 3-term one. The parameters of a CSS code are related to the homology group elements of the corresponding chain complex by the following.

\begin{proposition}\label{Proposition 3term complex}
Any 3-term complex $C:=
C_{i+1}\xrightarrow{\partial_{i+1}}C_i\xrightarrow{\partial_i}C_{i-1}$, given with a basis, defines a $[[ n, k, d]]$ CSS code $C$, with $n=\dim(C_i)$,
\begin{align*}
    k=&\dim(H_i(C))=\dim(H^i(C)),\\
    d=&\min \left \{|c|: [c]\in H_i(C)\sqcup H^i(C), [c]\neq0\right \}.
\end{align*} 
\end{proposition}
For an element of $c\in C_i$,  $[c]$ is the standard notation for the homology or cohomology class of $c$. The correspondence of Proposition \ref{Proposition 3term complex} is direct. Assuming that all vector spaces $C_i$ are given with a basis, the boundary maps can be interpreted as matrices. They play the role of the parity-check matrices $H_X$ and $H_Z$ or their transpose, and the linear codes are the subspaces $C_X=\operatorname{ker}(\partial_{i})$ and $C^{\perp}_Z=\operatorname{Im} (\partial_{i+1})$. \par

We end this section with an example of a CSS code, called hypergraph product codes (HPC) of Tillich and Z\'emor \cite{Tillich2009}, which will be studied thoroughly later in Section \ref{section Applications}. These quantum CSS codes are constructed from two classical codes, combined with the tensor product operation on chain complexes described in Section \ref{section chain complexes}.

\begin{definition}[Hypergraph product code]
    Let $C$ and $D$ be two classical codes. The hypergraph product code of $C$ and $D$ is the CSS code $C\otimes D^*:=(C_1\xrightarrow[]{\partial_1^C} C_0) \otimes ( D^1\xleftarrow[]{\delta_1^D}D^0 )$. 
\end{definition}
Explicitly, this CSS code is represented by the chain complex $C_1\otimes D^0\xrightarrow[]{\partial_2} (C_1\otimes D^1 \oplus C_0\otimes D^0 )\xrightarrow[]{\partial_1} C_0\otimes D^1$, with boundary maps $\partial_{i-j+1}|_{i,j}=\partial_i^ C\otimes \operatorname{id}_{D^*} + \operatorname{id}_C\otimes \delta_{j+1}^D$ for $i-j+1=1,2$. 

\subsection{Covering maps and fundamental group}

Both notions of lift of linear and quantum codes rely on the concept of covering maps from topology. We now review basic facts about coverings that we will need throughout this article, starting with Section \ref{section geometrical lifts of linear code}. Every result mentioned here can be found in \cite{HatcherTopo, Lyndon2001} and applies to topological spaces in general. However, for our usage, we only need to consider two-dimensional regular cell complexes, which are topological spaces obtained by successively gluing cells with gluing maps being homeomorphisms\footnote{Examples of gluing maps that are not homeomorphisms: the map $\partial \mathbb{I} \to p$, mapping both end-points of an interval to the same point, yielding a 1D-sphere; the map $\partial \mathbb{D}^2\to p$, mapping the boundary of the disk, a circle, to a point, yielding a 2D-sphere. }. For that reason, we skip certain topological definitions and refer to the texts above for more background.\par

We first recall the definition of a covering map. Given a topological space $K$, a \textit{covering} of $K$ consists of a space $K'$, together with a map $p:K'\to K$, such that for each point $x\in K$ there exists an open neighborhood $U$ of $x$, and a discrete space $S_x$, such that the inverse image of $U$ by $p$ is a disjoint union of open sets, $p^{-1}(U)=\bigsqcup_{i\in S_x} U_i$, where each set $U_i$ is called a \textit{sheet} and is mapped homeomorphically onto $U$ by $p$. Then, $K'$ is said to be a \textit{covering space} of the \textit{base space} $K$. The inverse image of a point $x$ by $p$ is the discrete space $S_x$ called the \textit{fiber} at $x$, and it is homeomorphic to the fiber at any other point. Its cardinality $|S_x|$ is hence the same for every point $x$ and called the $degree$ of the covering. A \textit{finite} covering map is one for which the degree is finite. A \textit{connected} cover is one for which $K'$ is a path connected space. A \textit{trivial covering } of $K$ is one for which $K'=\sqcup_{i} K_i $, where the restriction of $p|_{K_i}$ on each $K_i$ is a homeomorphism of $K$. \par

For example, if $p:K'\to K$ is a covering of a graph $K$, then $K'$ is a graph. A vertex $x'$ and its projection $x$ have the same degree. Moreover, an edge $e'$, with end-points $u'$ and $v'$, projects onto an edge $e$ with end-points $u=p(u')$ and $v=p(v')$.\par
A central property of covering maps that we will use in Section \ref{section Lift of a CSS code: definition} appears when we restrict the domain of covering map to certain subspaces.
\begin{lemma}\label{lemma restriction covering map}
    Let $p:K'\to K$ be a covering map and $A$ be a subspace of $K$. Let $A'=p^{-1}(A)$, the inverse image of $A$ by $p$ in $K'$. Then the restriction of $p$ to $A'$, namely, $p|_{A'}: A'\to A$ is a covering map.
\end{lemma}

A \textit{deck transformation} is an automorphism $d:K' \rightarrow K'$ such that $p\circ d=p$. The set of deck transformations, endowed with the operation of composition of maps, forms a group noted $\operatorname{Deck}(p)$. It is called the \textit{group of deck transformations} and acts on the left on $K'$. A \textit{regular covering} is a covering enjoying a left, free and transitive action of $\operatorname{Deck}(p)$ on the fiber. For any regular covering, $p:K'\rightarrow K$, it can be shown that $\operatorname{Deck}(p)\setminus K'\cong K$. \\

We now restrict to connected covering of \textit{well-behaved}\footnote{The literature on covering maps always starts by defining the notion of path-connected, locally path-connected, and semilocally simply-connected spaces. These are the requirement for classification results to apply. By well-behaved, we mean a space with these properties. In particular, any graph of finite degree and any finite 2D cell complex meet the requirements.} topological spaces. This is the type of spaces that we will consider later on. Moreover, all disconnected covering spaces can be obtained by disjoint union of connected ones, so we do not lose generality when we focus our study on connected spaces and their connected coverings, as we do now. There are a lot more results on coverings that can be stated in this context.\par
We first recall some definitions. Given a basepoint $v$ on a connected topological space $K$, the fundamental group $\pi_1(K,v)$ is the group of homotopy classes of loops based at $v$, which has for group operation the concatenation of loops. Later, we omit to mention the choice of basepoint and write it $\pi_1(K)$. As an example, for a connected graph $T=(V,E)$, the fundamental group is a free group of rank\footnote{The rank of a group $G$ is the smallest cardinality of a generating set of $G$.} $\#E-\#V+1$. A space is called \textit{simply connected} when its fundamental group is trivial. A result that we will need in Section \ref{section New constructions and lifts} is the Hurewicz Theorem, which relates homotopy and homology groups. For simplicity, we only state this theorem partially.
\begin{theorem}[Hurewicz]\label{theorem Hurewicz}
 Let $K$ be a path connected space. There exists an isomorphism \[ h:\pi_1(K)/[\pi_1(K),\pi_1(K)]\rightarrow H_1(K,\mathbb Z).\]
\end{theorem}
Here $[\pi_1(K),\pi_1(K)]$ denotes the commutator subgroup of $\pi_1(K)$. Hence, the domain of this map is the abelianization of $\pi_1(K)$.\par

The properties of covering maps over a space are closely related to that of its fundamental group. The next proposition will be crucial throughout our work.
\begin{proposition}
    Let $K$ be a well-behaved topological space. For every subgroup $H\leq \pi_1(K) $ there exists a covering $p:K_H\to K$, mapping basepoint to basepoint, and inducing an injective homomorphism $p_*:\pi_1(K',v')\to \pi_1(K,v)$, such that $p_*\pi_1(K_H , v') = H$.
\end{proposition}
Since the map $p_*$ is injective, this proposition also says that the fundamental group of $K_H$ is isomorphic to $H$. Most of the work of Section \ref{section Lift of a CSS code: explicit construction} consists in describing an explicit construction of $K_H$ for the cases of 1D and 2D cell complexes. Section \ref{section classification of lifts of HPC} also heavily makes use of this proposition. \par
A covering $p:K_H\to K$ is regular exactly when $H$ is a normal subgroup of $\pi_1(K,v)$. In that case, we also call this covering a \textit{normal covering}. It can be shown that, for such a covering map, we have the following isomorphism, $\operatorname{Deck}(p)\cong\pi_1(K,v)/H$.\par

The most important result on coverings, when we restrict to well-behaved spaces, is the classification theorem known as the Galois correspondence. This will be of central importance in Section \ref{section Lift of a CSS code: definition} and Section \ref{section classification of lifts of HPC}.
\begin{theorem}[Galois correspondence]\label{Theorem Galois correspond}
Let $K$ be a well-behaved topological space. There is a bijection between the set of basepoint-preserving isomorphism classes of path-connected covering spaces $p:(K',v')\to (K,v)$ and the set of subgroups of $\pi_1 (K, v )$, obtained by associating the subgroup $p_*\pi_1 (K',v' )$ to the covering space $(K',v')$. \par 
Given a subgroup $H\leq \pi_1(K)$ the degree $d$ of the covering is given by the index\footnote{Here, for groups $H\leq G$, $[G:H]$ is the standard notation for the index of $H$ in $G$.} $d=[\pi_1(K):H]$.
\end{theorem}
Ignoring basepoints, there is a bijection between isomorphism classes of path-connected covering spaces $p:K'\to K$ and conjugacy classes of subgroups of $\pi_1(K)$. When $K$ is simply connected, a consequence of Theorem \ref{Theorem Galois correspond} is that all of its coverings are trivial. Indeed, its connected covering must be the basepoint preserving homeomorphism onto itself, and all other covering spaces and maps can be obtained via disjoint union of connected ones.\par

As a result of this classification, we also have the following lemma that we will use in Section \ref{section classification of lifts of HPC}.
\begin{lemma}\label{Lemma covering of covering}
If $X_i$, $i=1,2$ are well-behaved connected coverings of $X$ associated to groups $G_i=p_{i*}\pi_1(X_i)$ and if $G_1\leq G_2$ then $X_1$ is a covering of $X_2$ of degree $[G_2:G_1]$.
\end{lemma}

To end this section, we mention the existence of a special type of covering of (well-behaved) connected spaces, called the \textit{universal covering}. This is the covering associated to the trivial subgroup of the fundamental group. Therefore, this is a normal covering and the associated covering space is unique up to homeomorphism. We will describe and use it in Section \ref{section relation to balanced product codes}.

\subsection{Linear codes, graphs and lift}\label{section geometrical lifts of linear code}

A classical linear code $C$ can be represented by its \textit{Tanner graph} $\mathcal{T}(C)$, which is a bipartite graph with one set of vertices representing bit variables and the other set representing check variables. There is an edge between a bit vertex and a check vertex when the bit is in the support of the check. In other terms, if the code is given by a parity check matrix $H$, the Tanner graph admits for adjacency matrix
\[
\renewcommand\arraystretch{1.3}
A=\mleft[
\begin{array}{c|c}
  0 & H  \\
  \hline
  H^T & 0 \\
\end{array}
\mright].
\]
Whenever it is clear from context, that we associate a certain Tanner graph to a code $C$, we write it $\mathcal{T}$ instead of $\mathcal{T}(C)$. \par
A lift \cite{Thorpe2003} is an operation of great interest to produce families of LDPC codes with dimension and distance linear in the length of the code, or with various other properties \cite{Pusane2011}. It corresponds geometrically to a covering of its Tanner graph. Initially, the notion of lift refers to general procedures applicable to any graph. Among them, permutation voltage and group voltage can generate respectively all possible graph coverings and all regular covers. The first appeared in \cite{Panteleev2020} as practical tool for the theory of quantum CSS codes. Here we give an account of group voltage, with a modification to the definition given in \cite{Panteleev2021}, and then permutation voltage following \cite{GROSS1977273}. We also show the effect of a lift on the Tanner graph of a linear code.\par

We start by setting the notation relative to graphs used throughout this article. Let $T:=(V,E)$ be a graph with set of edges $E$ and vertices $V$. An \textit{edge} of $E$ between two vertices $v$ and $v'$ is an unordered sequence of two vertices $\{u,v\}=\{v,u\}$. In this work, graphs have no self-loops, i.e edges of the form $\{u,u\}$. An \textit{oriented edge} is an edge for which the order matters, written $e=[u,v]$ for an edge running from $u$ to $v$. The inverse of an edge is defined as $[u,v]^{-1}:=[v,u]$. We note $E^*$ the set of oriented edges. Its cardinal is twice the one of $E$. A path in the graph is a sequence of oriented edges $\alpha=[u,v].[v,w]\dots $. \par
Next, let $\Gamma$ be a finite group. A \textit{voltage assignment}\footnote{The definition given here is different to the one in \cite{Panteleev2021}, which is a map $\nu:E\to G$, which must be given together with an orientation map.} with voltage group $\Gamma$ is a map
\begin{equation}\label{Equation voltage assignment}
     \nu:E^*\rightarrow \Gamma,
\end{equation} 
such that $\nu(e^{-1})=\nu(e)^{-1}$.  It is central to the following definition. 
\begin{definition}[Lift of a graph]\label{definition lift of graph}
  Let $T:=(V,E)$ be a graph and $\nu$ a voltage assignment. The \textit{right derived graph} $D(T,\mathcal{\nu})$ of $T$, also called \textit{right $\Gamma$-lift}, is a graph with set of vertices $V\times \Gamma$ and set of edges in bijection with $E\times \Gamma$, so that a vertex or edge, is written $(c,g)$ with $c\in T$ and $g\in \Gamma$. An oriented edge $(e,g)$, with $e=[u,v]$, in the graph $D(T,\mathcal{\nu})$, runs from $(u,g)$ to $(v, g \nu(e))$, where the multiplication by $\nu(e)$ is always on the right.
\end{definition}
With this definition of the voltage assignment, we have $(e,g)^{-1}=(e^{-1},g\nu(e))$. A left derived graph is similar, but the multiplication by $\nu(e)$ is on the left. In this work, we only consider right derived graphs unless otherwise stated. This is why we usually omit to mention it in the notation of the derived graph.\par

For every derived graph $D(T,\mathcal{\nu})$, it is possible to define a covering map $p:D(T,\nu)\rightarrow T$, which projects a cell $(c,g)\in D(T,\nu)$ to $c\in T$. It is also defined on the set $E^*$ of oriented edges, and preserves orientation. It was shown in \cite{GROSS1977273} that such group voltages can produce all regular (connected and disconnected) cover of $T$. For the right derived graph $D(T,\nu)$, the voltage group $\Gamma$ acts freely and transitively on the fiber, by left multiplication, and we even have $\operatorname{Deck}(p)\cong G$. Indeed, for any vertex, $u'=(u,g)$ we can define the left multiplication by $h\in \Gamma$ as $h.u'=(u,hg)$. Then, for an oriented edge $(e,g)=[(u,g),(v,g\nu(e))]$, multiplication by a group element $h$ satisfies $h(e,g)=(e,hg)$, since $(e,hg)=[(u,hg),(v,hg\nu(e))]$ is an edge in the same fiber as $(e,g)$.\par

\begin{remark}
For any left multiplicative action, there is an associated right group action, which is multiplication on the left by $g^{-1}$.
\end{remark}

Since any linear code can be represented by its Tanner graph, we can exploit the lift to define new codes.

\begin{definition}[lift of a linear code]\label{definition lift of linear code}
    Let $C$ be a linear code with Tanner graph $\mathcal{T}:=(A\cup B, E)$, where $A$ is the set of parity-check vertices and $B$ the set of bit vertices. Let $\nu:E^*\to \Gamma$ be a voltage assignment to some finite group $\Gamma$. When applying a right $\Gamma$-lift to the Tanner graph $\mathcal{T}$ we obtain a Tanner graph $D(\mathcal{T},\nu)$ with set of check and bit vertices respectively $A\times\Gamma $ and $B\times \Gamma$. We call the code $C^\Gamma$ associated to $D(\mathcal{T},\nu)$ a \textit{right $\Gamma$-lift }of $C$.
\end{definition}

Keeping notation as in Definition \ref{definition lift of linear code}, the code $C$ is identified with a chain complex $C_1\xrightarrow {\partial_1} C_0$. If we identify basis vectors in $C_0$ and $C_1$ with, respectively, the vertices of $A$ and $B$, we can define them by formal sums $C_0=\mathbb{F}_2A$ and $C_1=\mathbb{F}_2B$, with $\partial_1=H$. After applying a $\Gamma$-lift, the lifted code, noted $C^\Gamma$, has for Tanner graph $D(\mathcal{T},\nu)$. Therefore, $C^\Gamma$ has $|\Gamma|.\dim C_1$ bit variables and $|\Gamma|.\dim C_0$ parity-checks. We also refer to the resulting chain complex $C^\Gamma:=C_1^\Gamma\xrightarrow[]{\partial_{\Gamma,1}}C_0^\Gamma$ as a right $\Gamma$-lift of $C$. It follows from the property of covering maps, that $D(\mathcal{T},\nu)$ is locally homeomorphic to $\mathcal{T}$, so that the degree of each vertex in the derived graph $D(\mathcal{T},\nu)$ is equal to the degree of its projection in the base space $\mathcal{T}$. For that reason, the LDPC property of the code is conserved by the lift.\par

It is also well known from the theory of covering spaces that the chain complex of a normal cover with $\operatorname{Deck}(p)\cong \Gamma$ is a free $\mathbb{F}_2[\Gamma]$-module. While the chain complex of $C$ is not the cell complex of the graph $\mathcal{T}$, it nevertheless inherits the following module formulation.

\begin{proposition}
    Let $C:=C_1\xrightarrow {\partial_1} C_0$ be a linear code with Tanner graph $\mathcal{T}$ and $C^\Gamma$ be the $\Gamma$-lifted code associated to the Tanner graph $D(\mathcal{T},\nu)$. Then $C^\Gamma$ is identified with the chain complex

\begin{equation}
C^\Gamma:=\:\:\:\:C_1\otimes \mathbb{F}_2[\Gamma] \xrightarrow[]{\partial^\Gamma_1}C_0\otimes \mathbb{F}_2[\Gamma],
\end{equation} 
where the boundary map is defined on basis vectors by 
\begin{equation}
\partial^\Gamma_1(b\otimes g):= \sum_{a_i\in \operatorname{supp}\partial_1 b} a_i \otimes g\nu([b,a_i]), 
\end{equation} 
and extended by linearity over chains. 
\end{proposition}
Indeed, the tensor product is over $\mathbb{F}_2$ and $C_0=\mathbb{F}_2(A\times \Gamma) \cong \mathbb{F}_2A\otimes \mathbb{F}_2[\Gamma]$ and $C_1=\mathbb{F}_2(B\times \Gamma) \cong \mathbb{F}_2B\otimes \mathbb{F}_2[\Gamma]$. It is appropriate to identify any vertex $(v,g)$ in $D(\mathcal{T},\nu)$ with a basis vector $v\otimes g$ in $C_i^\Gamma$, $i=1,2$ and depending on the two formulations $\partial_1^\Gamma$ is represented either as a $|A|\times |B|$ matrix with coefficients in $\mathbb{F}_2 [\Gamma]$ or as a $|A|.|\Gamma|\times |B|.|\Gamma|$ matrix with values in $\mathbb{F}_2$, called the left regular representation. \\ 

We end this section by defining \textit{permutation voltages}. Let $n\in \mathbb{N}$ and $S_n$ be the symmetric group on the set of objects $X=\{1,\dots,n\}$. Let $T=(E,V)$ be an oriented graph. A permutation voltage on $T$ starts with a permutation-voltage assignment on the set of edges, $\mu:E^*\rightarrow S_n$, such that $\mu(e^{-1})=\mu(e)^{-1}$. The right derived graph $D(T,\mu)$ has set of edges $E\times X$ and set of vertices $V\times X$. A cell $c$, vertex or edge, is written $(c,i)$ for $i\in X$. An edge $(e,i)$ in the graph $D(T,\mu)$, where $e=[v,v']$, connects $(v,i)$ and $(v',\mu(e)^{-1}i )$. For any permutation voltage, it is also possible to define a covering map $p:D(T,\mu)\rightarrow T$ sending a cell $(c,i)\mapsto c$. The graph $D(T,\mu)$ is then a $n$-sheeted cover of $T$. It was also shown in \cite{GROSS1977273} that for any $n$-sheeted graph covering $p:T'\rightarrow T$, there exists a permutation voltage assignment of $T$ in $S_n$ such that the derived graph is isomorphic to $T'$. Therefore, the set of graphs obtained from group voltage on $T$  is a subset of the one obtained by permutation voltage. As before, applying this procedure to the Tanner graph of a linear code yields a lifted code in a natural way.

\begin{remark}
A lift of a graph obtained by group voltage with $\Gamma=S_n$ is a $n!$-sheeted regular cover, while a lift obtained by permutation voltage in $S_n$ is an arbitrary $n$-sheeted cover.
\end{remark}

\section{Lift of a quantum CSS code}\label{section Lift of a quantum CSS code}

\subsection{Tanner graph of a quantum CSS code}\label{section The Tanner cone-complex and lift}

Let $C:=\mathbb{F}_2Z\xrightarrow[]{\partial_2=H_Z^T} \mathbb{F}_2Q\xrightarrow[]{\partial_1=H_X} \mathbb{F}_2X$ be a quantum CSS code. Recall that we always fix a basis for each space and dual space. We begin with the basic definition of the Tanner graph of a CSS code. Throughout this article, graphs have no self-loops or multi-edges unless otherwise stated.

\begin{definition}[Tanner graph]\label{Quantum Tanner graph}
Let $C:=\operatorname{CSS}(C_X,C_Z)$ be a CSS code. The \textit{Tanner graph} of $C$ is the bipartite graph $\mathcal{T}(C)=((Z\cup X)\cup Q,E_{QZ}\cup E_{QX})$, in which $(Z\cup Q, E_{QZ})=\mathcal{T}(C_Z)$ and $(X\cup Q, E_{QX})=\mathcal{T}(C_X)$. Whenever it is clear from context, that we associate a Tanner graph to a code $C$, we write it $\mathcal{T}$ instead of $\mathcal{T}(C)$.
\end{definition}
That is, we now identify the sets of vertices in $\mathcal{T}(C)$ with the abstract cells in the chain complex $C$. For a check $s\in X\cup Z$, $\operatorname{supp}(s)$ denotes the support of the row vector corresponding to this check in the corresponding parity-check matrix. According to Definition \ref{Quantum Tanner graph}, this graph admits for adjacency matrix \par
\[
\renewcommand\arraystretch{1.3}
A=\mleft[
\begin{array}{c|c|c}
  0 & H_Z & 0 \\
  \hline
  H_Z^T & 0 & H_X^T\\
  \hline 
  0 & H_X & 0
\end{array}
\mright].
\] \par
The Tanner graph of a CSS code is an example of a bipartite graph $T$ with vertex set $V\cup V_0$, in which one of the subsets of vertices, say $V$, is subdivided into two sets, $V=V_1\cup V_2$, so that $T$ is also seen as a tripartite graph. For such a graph, we denote $E_{ij}$ the edge set between vertex sets $V_i$ and $V_j$. We recall that the distance between two vertices in a graph is defined as the number of edges in a shortest path connecting them. The following definition will appear to be central throughout this article.

\begin{definition}[Induced subgraph]\label{Definition induced subgraph}
    Let $T:=(V\cup V_0, E_{01}\cup E_{02})$, with $V=V_1\cup V_2$, be a bipartite graph. For any $u\in V_1$, let $T_u$ denote the subgraph of $(V_2\cup V_0, E_{02})$ composed of all the vertices of $V_0$ adjacent to $V_1$, all the vertices of $V_2$ of distance $2$ to $u$ in $T$, and all the edges of $E_{02}$ between these vertices. We name $T_u$ the \textit{subgraph induced} by $u$ in $(V_2\cup V_0, E_{02})$, or simply the subgraph induced by $u$. Exchanging the role of $V_1$ and $V_2$, for any $v\in V_2$, we call $T_v$ the subgraph induced by $v$ in $(V_1\cup V_0, E_{01})$.
\end{definition}
\begin{figure}[t]
  \centering
  \includegraphics[scale=0.9]{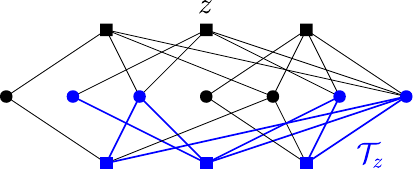}
  \caption{ Tanner graph of Steane's 7 qubit code \cite{Steane1996}. The upper row of vertices represents $Z$-checks, while the middle and lower rows are respectively qubits and $X$-checks. The induced subgraph of a $Z$-check $z$ is colored in blue.}
\end{figure}
We illustrate it for the Tanner graph $\mathcal{T}:=\mathcal{T}(C)$ with $C=\text{CSS}(C_X,C_Z)$. For any $z\in Z$, $\mathcal{T}_z$ denotes the subgraph of $\mathcal{T}({C_X})$ composed of all the qubits $q\in \text{supp}(z)$ and all the $X$-checks $x$ such that $\text{supp}(x)\cap \text{supp}(z)\neq \emptyset$. Then $\mathcal{T}_z$ is the subgraph induced by $z$ in $\mathcal{T}(C_X)$. Exchanging the role of $X$ and $Z$ checks, $\mathcal{T}_x$ the subgraph induced by $x$ in $\mathcal{T}(C_Z)$.

\begin{proposition}\label{proposition valid induced subgraph}
A bipartite graph of the form $T=(V\cup V_0, E_{01}\cup E_{02})$, with a partition $V=V_1\cup V_2$,  defines a valid Tanner graph for a CSS code iff for every $u\in V_1$ and every $v\in V_2$, $T_{v}$, respectively $T_u$, is a bipartite subgraph of even degree at the $V_2$ vertices, respectively the $V_1$ vertices.
\end{proposition}
\begin{proof}
We make the following identification: $V_0=Q$, $V_1=Z$, $V_2=X$. Then, this is equivalent to the relation $H_X.H_Z^T=0$.    \end{proof} 
\subsection{The Tanner cone-complex }\label{section The Tanner cone-complex}
In this section, we introduce a crucial element to lift a code $C:=\mathbb{F}_2Z\xrightarrow[]{\partial_2} \mathbb{F}_2Q\xrightarrow[]{\partial_1} \mathbb{F}_2X,$ which is a 2D geometrical complex that we call the \textit{Tanner cone-complex} of $C$. This is a canonical object which is associated to any CSS code. Later, a lift will be defined via a covering of this complex.\par
To describe any 2D complex, we first introduce the notion of a face, which is an element $\{v_1,v_2,v_3\}$ such that $[v_1,v_2].[v_2,v_3].[v_3,v_1]$ is a closed path in the graph.

\begin{definition}[Tanner cone-complex]\label{definition Tanner cone-complex 1}
Let  $\mathcal{T}(C)=(V,E)$, be the Tanner graph of the code $C=\text{CSS}(C_X,C_Z)$, with $V=X\cup Z\cup Q$ and $E=E_{QZ}\cup E_{QX}$. The \textit{Tanner cone-complex} $\mathcal{K}(C)$ is the 2D simplicial complex\footnote{Here, by simplicial complex, we mean the geometrical realization of the associated abstract simplicial complex. This is because we will consider its fundamental group and do geometrical operations. In the abstract definition, a simplicial complex is a set of sets of objects called "vertices",  which is closed under taking subsets. Every abstract simplicial complex defines a geometrical one, in which sets of 2 and 3 "vertices" are respectively edges and triangular faces. The converse is also true. } with 1-skeleton the graph $(V,E\cup E_{XZ})$, where \[E_{XZ}:=\{\{x,z\}, x\in X, z\in Z, \text{supp}(x)\cap \text{supp}(z)\neq \emptyset\},\]
and with the set of triangular faces $F:=\{\{x,q,z\}, q\in\text{supp}(x)\cap \text{supp}(z)\} $. In other words, there is at most a single edge between any pair of vertices corresponding to $(x,z)\in X\times Z$ when their support has non-empty intersection. When it is clear from context that a Tanner cone-complex is associated to a code $C$, we write it $\mathcal{K}$ instead of $\mathcal{K}(C)$.
\end{definition}
Let $I_{ZX}$ represent the $m_Z\times m_X$-matrix with rows indexed by the $X$-checks and columns by the $Z$-checks, such that $(I_{ZX})_{i,j}=1$ whenever the $i$-th $Z$-check and the $j$-th $X$-check have a common support. Then the 1-skeleton of $\mathcal{K}(C)$ admits for adjacency matrix
\[
\renewcommand\arraystretch{1.3}
A=\mleft[
\begin{array}{c|c|c}
  0 & H_Z & I_{ZX} \\
  \hline
  H_Z^T & 0 & H_X^T\\
  \hline
  I_{ZX}^T & H_X & 0
\end{array}
\mright].
\]Note that according to this definition, the Tanner cone-complex of any classical code $D$ satisfies $\mathcal{K}(D)=\mathcal{T}(D)$.\par

While the construction of the Tanner cone-complex given in \ref{definition Tanner cone-complex 1} is symmetric in $X$ and $Z$, we can also obtain it using the following method, which serves as an alternative definition. While Definition \ref{definition Tanner cone-complex 1} is very simple, Proposition \ref{Proposition Tanner cone-complex} highlights a property which will be central to show that the lift of a CSS code represents a valid code.

\begin{definition}\label{definition cone}
Let $X$ and $Y$ be spaces, $X$ being a subspace of $Y$, and $I$ denote the unit interval. The \textit{mapping cone }of the inclusion map $i:X\hookrightarrow Y$ is the space obtained by attaching the \textit{cone on} $X$, \[ \operatorname{C}X:=X\times I /(X\times \{0\}),\] to $Y$ according to the equivalence $(x,1)\sim i(x)$, $\forall x\in X$. It is noted $Y\cup_i \operatorname{C}X$.
\end{definition}

\begin{proposition}\label{Proposition Tanner cone-complex}
Let $C=\text{CSS}(C_X,C_Z)$. For any $z\in Z$, let $\mathcal{T}_z$ denote the subgraph induced by $z$ in $\mathcal{T}(C_X)$. Then, for all $z\in Z$ we can consider the mapping cones  $\mathcal{T}(C_X)\cup_{i_z} \operatorname{C}\mathcal{T}_z$ of the inclusion maps $i_z:\mathcal{T}_z\hookrightarrow \mathcal{T}(C_X)$, where $\operatorname{C}\mathcal{T}_z$ denotes the cone on $\mathcal{T}_z$. 
\begin{enumerate}
    \item The Tanner cone-complex is obtained by the attachment of all the cones on $\mathcal{T}_{z}$ for $z\in Z$, to $\mathcal{T}(C_X)$ at once, \[\mathcal{K}(C)\cong\mathcal{T}(C_X)\bigcup\limits_{z\in Z, i_z} \operatorname{C}\mathcal{T}_z  .\]
    \item Exchanging the role of $X$ and $Z$ checks, we also have \[\mathcal{K}(C)\cong\mathcal{T}(C_Z)\bigcup\limits_{x\in X, i_x} \operatorname{C}\mathcal{T}_x, \] where $\mathcal{T}_x$ is the subgraph induced by $x$ in $\mathcal{T}(C_Z)$.

\end{enumerate}
\end{proposition}

\begin{proof}
\begin{enumerate}
    \item We can label the vertex at the tip of $\operatorname{C}\mathcal{T}_z $ by $z$. Then, since every vertex $x$ in $\mathcal{T}_z$ gives rise to an edge in $\operatorname{C}\mathcal{T}_z $ when $\text{supp}(x)\cap \text{supp}(z)\neq \emptyset$, we obtain all the edges connecting $z$ with qubits on which it has support and vertices of $X$ which have overlapping supports. Moreover, each edge $[q,x]$ gives rise to a face $f$ which has for boundary the path of edges $[q,x].[x,z].[z,q]$. 
    \item This is direct from Definition \ref{definition Tanner cone-complex 1} of the Tanner cone-complex, which is symmetric in the sets of vertices $X$ and $Z$.\qedhere
\end{enumerate}
\end{proof}
\begin{figure}[t]
  \centering
  \includegraphics[scale=0.9]{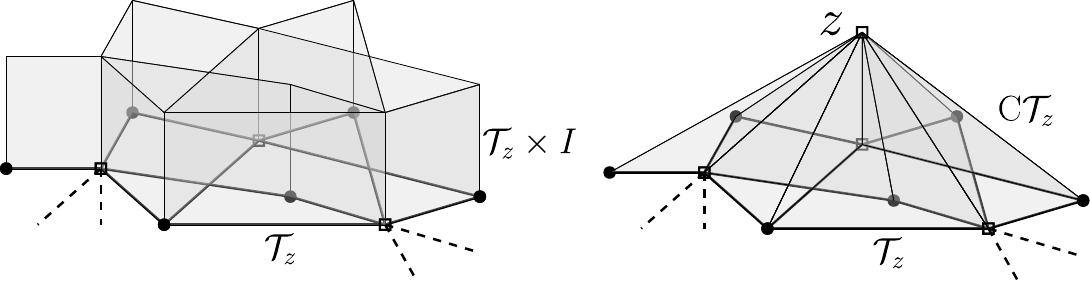}
  \caption{Cone on an induced subgraph $\mathcal{T}_z\hookrightarrow T(C_X)$, obtained by attaching $\mathcal{T}_z\times I$ to $\mathcal{T}(C_X)$ on one end, and then collapsing the other end to a point labeled $z$. Thick lines are edges of $\mathcal{T}_z$, dots are qubit vertices and boxes are check vertices. Dashed lines represent other edges in $\mathcal{T}(C_X)$. Thin lines are edges in $\mathcal{T}_z\times I$ or $\operatorname{C}\mathcal{T}_z$.}\label{Figure Tanner cone-complex}

\end{figure}
In Figure \ref{Figure Tanner cone-complex}, we illustrate the construction of the Tanner cone-complex of Proposition \ref{Proposition Tanner cone-complex}. Note that for any vertex $v$ representing a $X$ or $Z$-check, $\operatorname{C}\mathcal{T}_v$ is a contractible complex. This is a basic property of the cone, Definition \ref{definition cone} on a subcomplex. \par

\begin{remark}
The Tanner cone-complex has the following noticeable features.
   \begin{itemize}
       \item For any vertex $v$ representing a $X$ or $Z$ stabilizer, its link\footnote{For any vertex $v$ of a regular cell complex $K$, the link of $v$ in $K$ is defined as $\operatorname{link}(v,K):=\{ \sigma \in K, \sigma \cap v=\varnothing,\: v\cup \sigma\in X\}$, where by $v\cup \sigma$ we mean the cell defined by appending $v$ to the set of vertices defining $\sigma$.} is $\operatorname{link}(v)=\mathcal{T}_v$.
       \item For any vertex $q$ representing a qubit, $\operatorname{link}(q)$ is the complete bipartite graph composed of all the $X$ and $Z$ vertices having support on $q$ (and all the edges between them).
   \end{itemize}
\end{remark}
\begin{remark}
The Tanner cone-complex is the clique complex\footnote{The clique complex $K$ of a graph $T$ is a simplicial complex in which every $n$-vertex clique (every subset of $n$ vertices in which each pair of vertices are adjacent in $T$) of $T$ defines a $(n-1)$ dimensional cell in $K$} of its 1-skeleton. This is because any cycle of length 3 in its $1$-skeleton is specified by a triple $\{x,q,z\}$ with $q\in \text{supp}(x)\cap \text{supp}(z)$, which also defines a face in the Tanner cone-complex.
\end{remark}
This definition of Tanner cone-complex can be generalized to chain complex over $\mathbb{F}_2$-coefficients of any length. This is not needed in this work, but we state the definition for completeness. 
\begin{definition}
Consider the $\mathbb Z_2$-chain complex $C:=\: C_N\xrightarrow[]{\partial_N}\dots\xrightarrow[]{\partial_2} C_1\xrightarrow[]{\partial_1}C_0$, and denote the  $i$-th dual map $\delta_i$. The Tanner graph $\mathcal{T}(C)$ is the graph with set of nodes $V=\cup_0^N V_i$, where elements of $V_i$ are identified with basis vectors of $C_i$, and sets of edges $E=\cup_1^N E_{i,i-1}$, where there is an edge $e\in E_{i,i-1}$ between $u\in V_i$ and $v\in V_{i-1}$, iff $v\in \operatorname{supp}(\partial_{i} u)$. \par
The Tanner cone-complex $\mathcal{K}(C)$ is the 2D simplicial complex with set of edges $E\cup_2^N E_{i,i-2}$, where $E_{i,i-2}:=\{\{u,v\}, u\in V_i, v\in V_{i-2}, \text{supp}(\partial_i u)\cap \text{supp}(\delta_{i-1} v)\neq \emptyset\}$, and sets of faces $F=\cup_2^N F_{i,i-2}$, where $F_{i,i-2}:=\{\{u,v,w\}, u\in V_i, v\in V_{i-1},w\in V_{i-2},v\in\text{supp}(\partial_i u)\cap \text{supp}(\delta_{i-1}w).$\end{definition}

It is then possible to give an alternative construction using the cone on induced subgraphs, similarly to Proposition \ref{Proposition Tanner cone-complex}.

\subsection{Lift of a CSS code: definition}\label{section Lift of a CSS code: definition}

In this section, we define the lift of a CSS code. Intuitively, given a CSS code, one can consider covering maps of its Tanner cone-complex. It will first be shown that the associated covering space is the Tanner cone-complex of another CSS code. This new code will be defined as a lifted code. Then, the theory of covering maps can be used to classify all possible lifts.\par

The next proposition shows that a covering of any Tanner cone-complex can be used to define new quantum codes. This is not true in general if we try to do a covering of a Tanner graph directly, or of the 1-skeleton of the Tanner cone-complex. This is because for an arbitrary cover, the new induced subgraphs, see Definition \ref{Definition induced subgraph}, could be of odd degrees at certain check vertices. To understand this statement, consider two copies, noted $\mathcal{T}=(V,E)$ and $\mathcal{T}'=(V',E') $, of the Tanner graph of a CSS code $C$. Select in the first graph an arbitrary vertex $x\in X$, a qubit vertex $q$ adjacent to that check, and the copy of these vertices $x'$, $q'$, in the second graph. Then, consider the graph obtained from the disjoint copies, by adding two edges and removing two others, $T=(V\cup V',  \{ \{x,q'\},\{x',q\}\}\cup (E\cup E') \setminus \{ \{x,q\}, \{x',q'\}\} $. This is a 2-fold covering of the Tanner graph of $C$, obtained by a permutation voltage. However, for any vertex $z\in Z$ such that $x\in \mathcal{T}_z$ in $\mathcal{T}$, the new subgraph induced by $z$ in $G$ is of odd degree at $x'$, so that $T$ does not satisfy Proposition \ref{proposition valid induced subgraph}. We can repeat the same kind of argument, if we consider covering of the 1-skeleton of $\mathcal{K}(C)$. Therefore, to obtain a valid covering of a Tanner graph, another method is needed.

\begin{theorem} \label{Proposition Lift is valid CSS code}
Let $C$ be a CSS code with Tanner graph $\mathcal{T}$ and Tanner cone-complex $\mathcal{K}:=\mathcal{K}(C)$. Any finite covering map $p:K' \to \mathcal{K}$  of the Tanner cone-complex restricts to a covering map $p|_\mathcal{T}:T'\to\mathcal{T}$, for which the covering space $T'$ represents the Tanner graph of a quantum CSS code $C'$, $T'=\mathcal{T}(C')$, with sets of checks $X'=p^{-1}(X)$, $Z'=p^{-1}(Z)$ and qubits $Q'=p^{-1}(Q)$. 
\end{theorem}

\begin{proof}
We consider an arbitrary connected covering $p:K'\rightarrow \mathcal{K}$ of the Tanner cone complex. Since we have the embedding of the Tanner graph $\mathcal{T}\hookrightarrow \mathcal{K}$ as a subcomplex, and by Lemma \ref{lemma restriction covering map}, this covering map restricts to a covering map $p|_\mathcal{T}:T'\to\mathcal{T}$. We will show that $T'$ satisfies the property of Proposition \ref{proposition valid induced subgraph}. \par
Let $z\in Z$. The subcomplex $\operatorname{C}\mathcal{T}_z$ is simply connected, since by construction the cone on a subcomplex can be retracted onto a point. The restriction of a covering map to a subcomplex of $\mathcal{K}$ is also a covering of this subcomplex by Lemma \ref{lemma restriction covering map}. Hence, all the lifts of $\operatorname{C}\mathcal{T}_z$ in $K'$ must be disjoint copies of $\operatorname{C}\mathcal{T}_z$, that we denote as $\operatorname{C}\mathcal{T}_{(z,i)}$, where $i$ is an arbitrary indexing of the elements in the fiber above $z$. Therefore, each of these subcomplex $\operatorname{C}\mathcal{T}_{(z,i)}$, for $1\leq i\leq |S_z|$, corresponds to a new $Z$-generator $(z,i)\in Z\times S_z $, and it has as $\operatorname{link}((z,i))=\mathcal{T}_{(z,i)}$ which is a subgraph of the 1-skeleton of $\operatorname{C}\mathcal{T}_{(z,i)}$. Therefore, $\mathcal{T}_{(z,i)}\cong \mathcal{T}_z$, and in particular, it has even degree at the vertices of $X'$ corresponding to preimage by $p$ of $X$-vertices. We can repeat the same reasoning for all $x\in X$ and subcomplexes $\operatorname{C}\mathcal{T}_x$. As mentioned in Proposition\ref{proposition valid induced subgraph}, this is the only requirement for the 1-skeleton of $K'$ to describe a valid Tanner graph, once we remove edges between check vertices to give $T'=\mathcal{T}(C')$ for some code $C'$. 
\end{proof}

In the rest of this article, we only consider CSS codes with connected Tanner graphs, since a code with a disconnected Tanner graph is a disjoint union of several codes, each of which can be studied independently. This is also a requirement for the fundamental group $\pi_1(\mathcal{K}(C))$ to be uniquely defined, up to isomorphism\footnote{Given a basepoint $v$ on a topological space $K$, recall that the fundamental group is the group of homotopy classes of loops based at $v$, which has for group operation the concatenation of loops. For a space $K=\sqcup_i K_i$, and $v\in K_i$, $\pi_1(K,v)=\pi_1(K_i,v)$.}.

\begin{definition}[lift of a CSS code]\label{definition lift of quantum CSS code}
Let $C$ be a CSS code, with Tanner graph $\mathcal{T}$, and Tanner cone-complex $\mathcal{K}:=\mathcal{K}(C)$. Let $p:K'\rightarrow \mathcal{K} $ be a finite cover of $\mathcal{K}$ and $p|_\mathcal{T}:\mathcal{T'}\to\mathcal{T}$ its restriction to the Tanner graph. The \textit{lift }of $C$ associated to $p$ corresponds to the CSS code $C'$ such that $\mathcal{K}(C')=K'$ (or equivalently with Tanner graph $\mathcal{T}'=\mathcal{T}(C')$), sets of checks $X'=p^{-1}(X)$, $Z'=p^{-1}(Z)$ and qubits $Q'=p^{-1}(Q)$. A \textit{connected} lift of $C$ is a lift defined by a connected cover of $\mathcal{K}$. A \textit{trivial lift }of $C$ corresponds to disjoint copies of this code.
\end{definition}

When a lift of a code $C$ is disconnected, one only obtains a disjoint union of codes, each of which is obtained by a connected lift of $C$. For that reason, we restrict our attention to connected lifts. \par
\begin{remark}
Notice that the fundamental group $\pi_1(\mathcal{K}(C))$ depends on the choice of basis of $X$ and $Z$-checks. For example, for the Toric code this is $\mathbb Z\times \mathbb Z$, but if we replace a basis check $z\in Z$ by $z'+z$, with $z'\in Z$ such that $\operatorname{supp}(z)\cap \operatorname{supp}(z')=\emptyset$, then this becomes the free product $(\mathbb Z\times \mathbb Z)* \mathbb Z$ because, as seen in Proposition \ref{Proposition Tanner cone-complex}, we take the cone on a disconnected subgraph. \par 
\end{remark}
The condition for the existence of a non-trivial lift is the non triviality of $\pi_1(\mathcal{K}(C))$. When this condition is satisfied, the lift of a CSS code possesses the following properties.
\begin{proposition}\label{Proposition Lift properties}
The lift of a quantum CSS code, associated to a degree $d$ covering of its Tanner cone-complex, enjoys the following properties.
\begin{enumerate}
    \item         For an input CSS code of length $n$, it is a CSS code of length given by $d.n$.
    \item     The maximum weight of rows and columns of the lifted check matrices is unchanged compared to that of the input code.
    \item The dimension $k'$ of the new code is lower bounded by $k'\geq d.(|Q|-m_X-m_Z)$.
    \item If built from a finite regular cover, it has a free (linear) right action of the group of deck transformations and its dual complex has a similar left action. 
    \item     Applied to a classical code, it coincides with the geometrical lift of linear codes.
    \item Applied to a hypergraph product code \cite{Tillich2009}, it coincides with the lifted product code \cite{Panteleev2021}.
\end{enumerate}
\end{proposition}
\begin{proof}
   Items 1., 2. and $4.$ follow from directly the theory of covering spaces. Item 3 comes from the cardinality of the fiber in a covering map and from a simple counting argument on the number of qubits and checks in the lifted code. Item 5. comes from the fact that the Tanner cone-complex of a linear code is simply its Tanner graph. Item 6. will be the subject of Section \ref{section Applications}.
\end{proof}

\begin{remark}
We are not aware of any method to compute the dimension and distance of a lifted code in full generality, given the parameters of an input code. This is at least as difficult as determining the 1st homology and the systole of a surface, since the Tanner cone-complex of a surface code is a simplicial refinement of a cellulation. For that reason, the strategy adopted later is a case by case study.
\end{remark}
For example, several codes, such as Shor's \cite{Shor1995} and Steane's \cite{Stean1996Multiple} code, Quantum Reed-Muller codes\footnote{Quantum Reed-Muller codes are non LDPC, and one check has support on every qubit. As a general rule, taking cones over large subgraphs of a graph is likely to trivialize the fundamental group of the whole graph.}, and quantum CSS codes from finite geometry \cite{Audoux2015} cannot be lifted in a non-trivial way because their Tanner cone-complexes have a trivial fundamental group.  Some examples can be lifted in the usual sense: linear codes and codes coming from regular cellulations of surfaces. \par

Because the topological space $\mathcal{K}(C)$ is a connected cell complex, there is a way to classify all possible lifts of CSS codes. This is a direct consequence of the Galois correspondence, Theorem \ref{Theorem Galois correspond}. More precisely, for every subgroup $H$ of $\pi_1(\mathcal{K},v)$ there exists a path-connected covering space $\mathcal{K}'$ such that the $p_*\pi_1 (\mathcal{K}',v' )\cong H$. The new complex $\mathcal{K}'$ is the Tanner cone-complex of a lifted code $C'$, by Proposition \ref{Proposition Lift is valid CSS code}. This statement of existence is essential, but we will not be interested in every subgroup of the fundamental group to create codes. Only subgroups of finite index in $\pi_1(\mathcal{K})$ will give codes with a finite number of qubits.\par
Using this correspondence, it can be appropriate to denote $C_H$ the $H$-lift of a CSS code $C$. This is in contrast with the notation of Definition \ref{definition lift of linear code}, where the $\Gamma$-lift of a linear code $C$ is written $C^\Gamma$, for an arbitrary choice of finite group $\Gamma$, playing the role of the fiber of the covering. For a lifted CSS code, whenever we wish to put emphasis on the subgroup $H \leq\pi_1(\mathcal{K})$, we write $H$ in subscript. This notation is suitable when, for instance, the subgroup isn't normal. However, when $H\trianglelefteq \pi_1(\mathcal{K})$ and we wish to put emphasis on the fiber of the covering $\Gamma=\pi_1(\mathcal{K})/H$, we write the lifted CSS code $C^\Gamma$. This is useful when only a complicated presentation of $H$ is available, while the quotient group $\Gamma$ is well known.\par

\subsection{Lift of a CSS code: explicit construction}\label{section Lift of a CSS code: explicit construction}

In this section, we describe the procedure to construct the lift of a quantum CSS code explicitly. To this end, we first show how to produce a covering of the 1-skeleton of the Tanner cone-complex, which can be completed into a covering of the full complex. Since all the information about the lifted quantum code is contained in the 1-skeleton, we can derive the parity-check matrices of the lifted code from this graph covering. We keep the notations of graphs established in Section \ref{section geometrical lifts of linear code}. Most of the material on covering maps and fundamental group is adapted from \cite{Lyndon2001,HatcherTopo}, therefore the construction is detailed but proofs usually omitted.\par

We consider the CSS code $C:=\mathbb{F}_2Z\xrightarrow[]{\partial_2} \mathbb{F}_2Q\xrightarrow[]{\partial_1} \mathbb{F}_2X$ and denote the dual maps $\delta_i$ for $i=1,2$. Let $\mathcal{T}:=\mathcal{T}(C)$, and $\mathcal{K}:=\mathcal{K}(C)$. Recall that $\mathcal{K}^1=(V,E)$, the 1-skeleton of $\mathcal{K}$, is not isomorphic to $\mathcal{T}$ since $\mathcal{K}^1$ has edges between the $X$ and $Z$ vertices. \par
The topological space $\mathcal{K}$ is a finite connected simplicial complex. From the Galois correspondence, Theorem \ref{Theorem Galois correspond}, connected coverings of the base space $\mathcal{K}$, are in one-to-one correspondence with subgroups of its fundamental group $\pi_1(\mathcal{K},q_0)$, where $q_0$ is an arbitrary choice of basepoint in $\mathcal{K}$ that we now omit. In particular, let $H$ be a subgroup of $\pi_1(\mathcal{K})$. There exists a unique covering map, 
\begin{equation}\label{equation covering}
p_H:\mathcal{K}_H \rightarrow \mathcal{K},
\end{equation}
associated to $H$. It has degree $d=[\pi_1(\mathcal{K}):H]$. The lifted code $C_H$ is the one which satisfies $\mathcal{K}(C_H)=\mathcal{K}_H$, with the assignment of checks and qubits to vertices prescribed in Definition \ref{definition lift of quantum CSS code}. \par
We will first construct the covering of Equation \eqref{equation covering}, by starting with a covering of the 1-skeleton of $\mathcal{K}$. To do so, we have to consider the following subgroups and homomorphisms,
\[
\begin{tikzcd}[column sep=0cm, row sep=0.6cm]
  &  \pi_1(\mathcal{K}) 
  & \geq 
  &   H \arrow{d}{\phi^{-1}}
\\
  & \pi_1(\mathcal{K}^1) \arrow{u}{\phi} 
  & \geq
  & H^1
\end{tikzcd},
\]
where $\phi$ is the natural homomorphism, taking the quotient of $\pi_1(\mathcal{K}^1)$ by the normal closure\footnote{Let $\langle S\:|\:R\rangle$ be a presentation of a group $G$. The normal closure $N$ of $R$ in $F_S$, the free group generated by $S$, is the smallest normal subgroup of $F_S$ containing $R$, $N=\langle grg^{-1}: g\in F_S, r\in R\rangle $. Then, we have $\langle G=F_S/N$.} of the subgroup generated by relations coming from loops which bound a face in $\mathcal{K}$. Moreover, $H^1$ is the preimage of $H$ in $\pi_1(\mathcal{K}^1)$ by this homomorphism. We now use the shorthand notations, $\Gamma:=\pi_1(\mathcal{K} )/H$ and $\Gamma^1:=\pi_1(\mathcal{K}^1)/H^1$. Whenever $H$ is normal, there is an isomorphism $\Gamma\cong\Gamma^1$. Otherwise, this is a bijection between the cosets, after specifying whether the cosets are taken from the left or right. Here, all cosets are taken on the right. \par
To construct the covering map $p_{H}:\mathcal{K}_H\rightarrow \mathcal{K}$, we must begin with a construction of the covering map \begin{equation}\label{Equation covering of 1-skeleton}
p^1_{H^1}:\mathcal{K}^1_{H^1}\rightarrow \mathcal{K}^1,
\end{equation}
which exists and is also unique by Theorem \ref{Theorem Galois correspond}. It can be obtained as the lift of a graph as exposed in Definition \ref{definition lift of graph}, but with a choice of voltage assignment specifically designed to produce connected coverings. We explain the steps to create $p^1_{H^1}$ and $\mathcal{K}^1_{H^1}$, the latter being a graph with sets of vertices $Q\times \Gamma^1\cup X\times \Gamma^1\cup Z\times \Gamma^1$ that we now describe how to connect by edges. Since $\mathcal{K}^1$ is a connected graph, it admits a maximal spanning tree $T$ with basepoint $q_0$, and by definition a unique path from vertex to vertex. We define a specific type of voltage assignment on the set of oriented edges,
\begin{equation}\label{Equation modified voltage assignment}
    \nu^T: E^* \to \pi_1(\mathcal{K}^1)
\end{equation}
mapping an edge $e=[u,v]$ to an element of $\pi_1(\mathcal{K}^1)$ in the following manner. Let a path $\alpha$ from vertex $a$ to $b$ be denoted as a sequence of oriented edges, for example $\alpha=[a,u_1].[u_1,u_2]\dots [u_j,b]$. Moreover, let $((a,b))$ denote the path from $a$ to $b$ in the tree $T$. Then the equivalence class of loops obtained by adding $e$ to $T$ is $\nu^T(e):=[((q_0,u)).e.((v,q_0))]$, where $[\alpha]$ is the standard notation for the homotopy class of a loop $\alpha$ based at $q_0$. When we have an edge $[u,v]$, we note $\nu^T[u,v]:=\nu^T([u,v])$, to make the notation less cluttered. We define a new construction of derived graphs associated to this voltage assignment.

\begin{definition}[Connected lift of a graph]\label{definition modified lift of graph}
Let $G=(V,E)$ be a connected graph, $T$ a spanning tree of $G$, $\nu^T:E^*\to\pi_1(G)$ a voltage assignment on oriented edges as described above, $H$ a subgroup of $\pi_1(G)$ and let $\Gamma:=\pi_1(G)/H$. The right-derived graph $D(G,\nu^T,H)$ is a graph with set of vertices $V\times \Gamma$ and set of edges $E\times \Gamma$, so that a vertex or edge is written $(c,gH)$ with $c\in G$ and $g\in \pi_1(G)$. An oriented edge $e'=(e,gH)$ in the graph $D(G,\nu^T,H)$, with $g\in \pi_1(G)$ and $e=[u,v]$,  connects $(u,gH)$ and $(v, g \nu^T(e)H)$, where the multiplication by $\nu^T(e)$ is always on the right. 
\end{definition}
In general, a graph admits more than one spanning tree. Following this procedure with another choice of spanning tree will modify the parametrization of the vertices and edges, but the graphs will be isomorphic, by the Galois correspondence. \par
When $H^1\trianglelefteq \pi_1(\mathcal{K}^1)$, this definition of  lifts coincides with the one given in Definition \ref{definition lift of graph} for the specific choice of voltage assignment given here. But while Definition \ref{definition lift of linear code} can generate all (connected and disconnected) regular covers, Definition \ref{definition modified lift of graph} can generate all connected regular and non-regular covers.\par

With Definition \ref{definition modified lift of graph} of right derived graph, we claim that $\mathcal{K}^1_{H^1}=D(\mathcal{K}^1,\nu^T,H^1)$. For that, we define $p^1_{H^1}$ as the map which projects any cell of $\mathcal{K}^1_{H^1}$, vertex or edge, to its first coordinate. This is a valid covering map, equivalent to what we have done in Section \ref{section geometrical lifts of linear code}. It is also possible to prove that $p^1_{H^1}$ induces an isomorphism $p^1_{H^1*}:\pi_1(\mathcal{K}^1_{H^1},q_0')\to H^1$, for any choice of basepoint $q_0'$ in the preimage of $q_0$, by studying the image of loops in the cover (see \cite{Lyndon2001}, Chapter III. 3). Therefore, the graph $\mathcal{K}^1_{H^1}=(\{Q\times \Gamma^1,X\times \Gamma^1,Z\times \Gamma^1\}, E\times \Gamma)$  and the map $p^1_{H^1}$ define the covering of Equation \eqref{Equation covering of 1-skeleton}. \par

We could describe how to complete the covering of Equation \eqref{Equation covering of 1-skeleton} into the one of Equation \eqref{equation covering} by geometrical means. But this step is not required to derive the boundary maps of the lifted code. Their expression is summarized in the next proposition.

\begin{proposition}\label{Proposition boundary maps of lifted code}
Let $C:=\mathbb{F}_2Z\xrightarrow[]{\partial_2} \mathbb{F}_2Q\xrightarrow[]{\partial_1} \mathbb{F}_2X$ be a CSS code, and denote the dual maps $\delta_i$ for $i=1,2$. Let $\mathcal{K}:=\mathcal{K}(C)$, $T\subseteq \mathcal{K}^1=(E,V) $ a spanning tree of its 1-skeleton, $\nu^T: E^* \to \pi_1(\mathcal{K}^1)$ a voltage assignment on oriented edges, and $\phi: \pi_1(\mathcal{K}^1)\to\pi_1(\mathcal{K})$ the natural homomorphism. For a given choice of subgroup $H\leq \pi_1(\mathcal{K})$ with finite index, let $C_H$ be the associated lifted CSS code of $C$. We write $\Gamma$ the set of $H$-coset in $\pi_1(\mathcal{
T})$. Then, this code can be written as
\[
 C_H:=\mathbb{F}_2^{m_Z.|\Gamma|} \xrightarrow[]{\partial_{H,2}}\mathbb{F}_2^{n.|\Gamma|}\xrightarrow[]{\partial_{H,1}} \mathbb{F}_2^{m_X.|\Gamma|}.
\]
The boundary maps of the lifted code and its dual (in the sense of chain complex) are defined on basis elements of $C_H$ by
\begin{equation}\label{equation parity-checks of lift 1}
\begin{aligned}
     \delta_{H,1}(x,gH)=&\sum_{q\in \text{supp}(\delta_1x)} (q,g\phi(\nu^T[x,q])H),\\
\partial_{H,1}(q,gH)=&\sum_{x\in \text{supp}(\partial_1 q)}(x,g\phi(\nu^T[q,x])H),
\end{aligned}
\end{equation}
where $g\in \pi_1(\mathcal{K})$ is any group element. Similarly, we have
\begin{equation}\label{equation parity-checks of lift 2}
\begin{aligned}
     \delta_{H,2}(q,gH)=&\sum_{z\in \operatorname{supp}(\delta_2q)} (z,g\phi(\nu^T[q,z])H),\\
\partial_{H,2}(z,gH)=&\sum_{q\in \operatorname{supp}(\partial_2 z)}(q,g\phi(\nu^T[z,q])H).
\end{aligned}
\end{equation}
Their action is extended by linearity over chains.
\end{proposition}
\begin{proof}
The dimension of each vector space of the lifted code $C_H$ is multiplied by the degree of the cover $p:\mathcal{K}_H\to\mathcal{K}$ by Proposition \ref{Proposition Lift properties}; they are $m_Z.|\Gamma|$ for the $Z$-checks, $m_X.|\Gamma|$ for the $X$-checks and $n.|\Gamma|$ the qubits.\par
To express the boundary maps, we recall that the graph $\mathcal{K}^1_{H^1}$, with $H^1:=\phi^{-1}(H)$, carries all the data needed to construct the $H$-lifted code of $C$, since the new Tanner graph is just a subgraph of the 1-skeleton. We can determine new incidence matrices, $(\delta_{H,i},\partial_{H,i}), i=1,2$, simply from the adjacency of vertices in $\mathcal{K}^1_{H^1}$. For instance, for $i=1$ they are written as follows on basis vectors, $ \delta_{H,1}(x,gH^1)=\sum_{q\in \text{supp}(\delta_1x)} (q,g\nu^T[x,q]H^1)$ and $\partial_{H,1}(q,gH^1)=\sum_{x\in \text{supp}(\partial_1 q)}(x,g\nu^T[q,x]H^1)$, where $g\in \pi_1(\mathcal{K}^1)$. However, since we start from a choice of group $H$, it is more adequate to use another parametrization of the fiber which depends on $H$. For this, we make use of the natural homomorphism $\phi$. \par 
Proposition \ref{Proposition Lift is valid CSS code} ensures that the lifted boundary maps represent valid parity-check matrices for $C_H$, and that $\mathcal{K}_H=\mathcal{K}(C_H)$. \qedhere
\end{proof}

Finally, the total space of $p_H:\mathcal{K}_H \rightarrow \mathcal{K}$ can be reconstructed from the definition of the Tanner cone-complex, with $\mathcal{K}_H=\mathcal{K}(C_H)$. The map $p_H$ then sends a face $\{x',q',z'\}$ of $\mathcal{K}_H$, to $\{p_H(x'),p_H(q'),p_H(z')\}$. This ends the definition of the covering map.

\subsection{Relation to the fiber bundle code}\label{section relation to fiber bundle codes}

We have detailed a method to lift a CSS code $C$ via a connected covering of its Tanner cone-complex. We now explain how this is related to the fiber bundle code construction of \cite{Hastings2020}.

The fiber bundle formulation of a lifted code is closely related to what we have done so far, but we need to restrict to normal covering maps and translate our formula to highlight the similarities. A normal covering map is also a principal $G$-bundle $\mathcal{F}\hookrightarrow\mathcal{E}\to \mathcal{B}$, over a base space $\mathcal{B}$, with total space $\mathcal{E}$, fiber $\mathcal{F}$, a discrete group, on which a group $G\subset \text{Aut}(\mathcal{F})$ acts freely and transitively. The group $G$, called the structure group, is the group of automorphisms of the fiber. In our case, we consider covering maps over the base space $\mathcal{K}:=\mathcal{K}(C)$. The fiber is the finite group $\Gamma:=\pi_1(\mathcal{K})/H$, which is also the structure group. \par

Taking inspiration of principal $G$-bundles, we can form an algebraic analogue, where the base space and fiber are replaced by two different chain complexes. The chain complex of the base space is the complex of the code $C$. The fiber is the group ring $F_0=\mathbb{F}_2[\Gamma]$, which may be extended into the trivial linear code $F:=F_0\xrightarrow{\partial_0^F}0$. The action of $\Gamma$ on the fiber induces a linear action on the vector spaces of $F$, since we have trivially $\gamma \partial_0^F f=\partial_0^F \gamma f$, for any chain $f\in F_0$, and any $\gamma\in \Gamma$.\par

The Tanner cone-complex $\mathcal{K}_H$, and hence the code $C_H$, inherits the original free transitive linear left action of the group of deck transformations of the covering $\operatorname{Deck}(p_H)\cong \Gamma$\footnote{we recall that the group of deck transformations is the group of homeomorphism with elements $d:\mathcal{K}_H\to \mathcal{K}_H$ such that $p_H\circ d=p_H$} and is hence a free $\Gamma$-module.  Its dual complex $C^{\Gamma*}$ enjoys a similar right action of $\operatorname{Deck}(p_H)$. Taking advantage of that, it is possible to define a fiber bundle code \cite{Hastings2020}, identified with the chain complex $C_H=(C_{H,\bullet},\partial^E_\bullet)$, for which the chain spaces are written
\begin{equation}\label{equation fiber bundle code 1}
  C_{H,n}=\oplus_{r+s=n} C_r\otimes F_s.  
\end{equation}
This is not the cellular chain complex\footnote{The cellular chain complex of a cell complex $K$ is the chain complex indexed by dimension of cells, and whose $n$-chains are formal combination of $n$-dimensional cells, and boundary maps defined on a cell by formal sums of its boundary cells.} of $\mathcal{K}_H$, which would have cells indexed by order of dimension. Equation \eqref{equation fiber bundle code 1} is a simple rewriting of a basis element $(q,Hg)$ into $q\otimes Hg$. To express the new boundary map $\partial_H$ we have to specify a chosen connection of the bundle, which is an assignment of a fiber automorphism element  $\varphi:\{(u,v), v\in \operatorname{supp}(\partial_i u) \} \rightarrow\text{Aut} (\mathcal{F})$. Then, it is possible to make the parallel with the code of \cite{Hastings2020} more precise.

\begin{proposition}
    Suppose $H$ is a normal subgroup of $\pi_1(\mathcal{K})$. Then, the $H$-lift of $C$ is identified with the fiber bundle code \cite{Hastings2020} $C_{H}=(C_{H\bullet},\partial_{H,\bullet})$, where $C_{H,n}=\oplus_{r+s=n} C_r\otimes F_s$. A connection is given on a pair $(u,v)\in C_{H,n}$ by $\varphi(u,v) f:=f \nu^T[u,v]$. On basis elements, the boundary maps read
\begin{equation}\label{equation fiber bundle}
\begin{aligned}
    \partial_{H,2}(z\otimes f)=&\sum_{q\in \operatorname{supp} (\partial_2 z)}q\otimes \varphi(z,q) f,\\
    \partial_{H,1} ( q\otimes f)=&\sum_{x\in \operatorname{supp}(\partial_1 q)} x\otimes \varphi(q,x) f.
\end{aligned}
\end{equation}
\end{proposition}
\begin{proof}
    Equation \eqref{equation fiber bundle} is a particular case of Equation 6 of \cite{Hastings2020}. They are merely a rewriting of  the boundary maps given in Equations \eqref{equation parity-checks of lift 1} and \eqref{equation parity-checks of lift 2}.
\end{proof}

\subsection{Relation to balanced product codes}\label{section relation to balanced product codes}
In this section, we show how the lift is related to a certain balanced product operation, drawing a parallel with the construction in \cite{Breuckmann2021}. This is inspired by the algebraic approach to local systems of coefficients \cite{Davis2001LectureNI,HatcherTopo}, which already found applications in linear coding \cite{meshulam2018}. Usually, one defines a local coefficient system over the cellular chain complex of a space $K$, with the help of a covering map. Here, we slightly revisit the procedure to define an analogue over an input code $C:=C_2\xrightarrow{\partial_2}C_1\xrightarrow{\partial_1}C_0$. Contrary to Section  \ref{section relation to fiber bundle codes}, this formulation is not restricted to the case of normal covering maps. In this section, we denote $\pi_1:=\pi_1(\mathcal{K})$.\par

The Tanner cone-complex $\mathcal{K}:=\mathcal K (C)$ is a path connected space and has a universal cover, $p:\widetilde{\mathcal{K}}\to \mathcal{K}$ which coincides with the quotient map $\widetilde{\mathcal{K}} \to\pi_1 \setminus\widetilde{\mathcal{K}}$, where $\pi_1$ acts on the left by deck transformation. This complex is related to a specific code.

\begin{definition}
Consider the CSS code $C:=\:\mathbb F_2Z\xrightarrow{\partial_2}\mathbb F_2Q\xrightarrow{\partial_1}\mathbb F_2X$, with $\mathcal{K}:=\mathcal K(C)$ having universal cover $p:\widetilde{\mathcal{K}} \to\pi_1 \setminus\widetilde{\mathcal{K}}$. We call the \textit{universal lift} of $C$, the chain complex $\widetilde C:=\:\mathbb F_2\widetilde Z\xrightarrow{\partial_2}\mathbb F_2\widetilde Q\xrightarrow{\partial_1}\mathbb F_2\widetilde X$, such that $\widetilde{\mathcal{K}}=\mathcal{K}(\widetilde{C})$, $\widetilde X=p^{-1}(X)$, $\widetilde Z=p^{-1}(Z)$ and $\widetilde Q=p^{-1}(Q)$.
\end{definition}
Notice that we are careful not to call $\widetilde C$ a code, unless each of its spaces has finite dimension. For the Toric code, the universal lift is the cellular chain complex of a cellulation of $\mathbb R^2$ by finite-size squares.
   
\begin{proposition}\label{proposition universal lift of code as module}
    For a CSS code $C$, its universal lift $\widetilde{C} $ is a chain complex of left $\mathbb{F}_2[\pi_1]$-modules.
\end{proposition}
\begin{proof}
From the theory of covering map, the free transitive left action of $\pi_1$ on $\widetilde{\mathcal{K}}$ induces a similar action of $\pi_1$ on the $\mathbb F_2$-cellular chain complex of $\widetilde{\mathcal{K}}$ (seen as a 2D simplicial complex), $\widetilde S:=S(\widetilde{ \mathcal K},\mathbb F_2)$, endowing it with the structure of a  module over the group algebra $\mathbb{F}_2[\pi_1]$. The boundary maps of $\widetilde S$, noted $\widetilde{d}_{0\leq i\leq 2}$, and its coboundary maps $\widetilde{d}^{0\leq i\leq 2}$, become $\mathbb{F}_2[\pi_1]$-module homomorphisms. This action of $\pi_1$ on $\widetilde S$, in turn, induces an action of $\pi_1$ on each vector space $\widetilde{C}_{0\leq i\leq 2}$, defined by restricting the action of $\pi_1$ to the vertices of $\widetilde S_0$, interpreted as abstract cells of $\widetilde{C}$. This makes each space $\widetilde{C}_{0\leq i\leq 2}$ into a $\mathbb{F}_2[\pi_1]$-module. The boundary maps $\widetilde{\partial}_i: \widetilde{C}_i\to \widetilde{C}_{i-1}$ are moreover $\mathbb{F}_2[\pi_1]$-module homomorphisms. Indeed, they can be constructed from the boundary maps of $\widetilde S$ in the following way. For $r\in \mathbb{F}_2[\pi_1]$ and $z\in \widetilde Z$, we can interpret the cell $rz$ both as a cell of $\widetilde S$ or one of $\widetilde C_2$. Then the map $\tilde \partial_2$ satisfies $\tilde \partial_2 (r z)=(\tilde d_1 \tilde d^0  r z) |_Q= (r \tilde d_1 \tilde d^0 z)|_Q$ where $|_Q$ is the restriction map to the vertices representing cells (qubits) of $\widetilde C_1$, which is also a $\mathbb{F}_2[\pi_1]$-module homomorphism. Therefore $\tilde \partial_2 (r z)= r\tilde \partial_2  z$, making it a module homomorphism. The boundary map defined on basis vector by $\tilde \partial_1 q=(\tilde d_1 \tilde d^0 q) |_X $ has the same property. The coboundary maps of $\widetilde C$ can be obtained similarly.\qedhere
\end{proof}

Let $A$ be a discrete Abelian group over which the group $\pi_1$ acts on the left via a homomorphism $\rho:\pi_1\to \operatorname{Aut}(A)$. This endows $A$ with the structure of a left $\mathbb{Z}[\pi_1]$-module. Applying the tensor product with $ \mathbb F_2$, seen as a $\mathbb Z$-module, yield a $\mathbb{F}_2[\pi_1]$-module, noted $A_\rho$. For a cell complex $K$, a local coefficient system on $S:= S( K,\mathbb F_2)$ is defined as a complex of modules and morphisms $\widetilde S \otimes_{\mathbb{F}_2[\pi_1]} A_\rho:=(\widetilde{S}_\bullet\otimes_{\mathbb{F}_2[\pi_1]} A_\rho, \widetilde{d}_\bullet \otimes \mathbb{1})$. The tensor product $\otimes_{\mathbb{F}_2[\pi_1]}$ of a left and a right $\mathbb{F}_2[\pi_1]$-module has the effect of operating a change of ring. We now replace $K$ with $\mathcal K (C)$, and consider a similar construction for the lifts of $C$ giving them a balanced product formulation.

\begin{proposition}
 Let $C$ be a CSS code with Tanner cone-complex $\mathcal{K}$, $H$ be a subgroup of $\pi_1$ and $\Gamma=\pi_1/H$, the left cosets. Then we have an isomorphism of chain complexes,
\begin{equation}\label{equation local coeff}
C_{H} \cong \widetilde{C}\otimes_{\mathbb{F}_2[\pi_1]} \mathbb{F}_2[\Gamma].
\end{equation}
\end{proposition}

\begin{proof}
We keep the notation as in the proof of Proposition \ref{proposition universal lift of code as module}. \par
In the description of local systems, we take $A=\mathbb{F}_2[\Gamma]$. This is an Abelian group, although $H$ is not necessarily a normal subgroup (in that case it corresponds to formal sums of coset elements). The ring $A$ can be endowed with the structure of a left $\mathbb{F}_2[\pi_1]$-module, with the natural action of the fundamental group on the coset space. From the Galois correspondence $H$ is associated to a covering space $p: \mathcal{K}_H\rightarrow \mathcal{K}$ and $H=p_\star\pi_1(\mathcal{K}_H)$. For the present choice of $A$, it can be shown that $\widetilde S \otimes_{\mathbb{F}_2[\pi_1]} A_\rho\cong S(\mathcal K_H,\mathbb F_2)$ (see \cite{HatcherTopo}, Chapter 3.H). \par
Moreover, the left modules $\widetilde{C}_{0\leq i\leq 2}$ can be seen as right modules with the action $\sigma . g:=g^{-1} \sigma $ on abstract chains $\sigma$. Therefore, the space $\widetilde{C}_i\otimes_{\mathbb{F}_2[\pi_1]} A_\rho$ is well-defined and can be both seen as a $\mathbb{F}_2[\pi_1]$-module and a $\mathbb{F}_2[\Gamma]$-module. \par
Each module of $\widetilde S \otimes_{\mathbb{F}_2[\pi_1]} A_\rho$, and in particular the $S_0(\mathcal K_H,\mathbb F_2)$ has a direct sum decomposition in terms of fibers $ p^{-1}(v)$ of the covering map. As shown in the proof of Proposition \ref{proposition universal lift of code as module}, $\widetilde C_i$ is constructed from $\widetilde S_0$ by interpreting the cells of $\widetilde S_0$ as abstract cells of $\widetilde C_i$, i.e. each fiber $ p^{-1}(v)$, where $v\in C_i$ is given a unique type ($Z$, $Q$ or $X$), depending on the type of $v$.  Therefore, we can also reconstruct each space $\widetilde{C}_i\otimes_{\mathbb{F}_2[\pi_1]} A_\rho$ by the span of the cells in $ p^{-1}(v)$, for $v$ interpreted as cells of $C_i$, making it a subspace of $\widetilde S_0 \otimes_{\mathbb{F}_2[\pi_1]} A_\rho$ stable under the action of $\mathbb{F}_2[\pi_1]$. Then we form the chain complex $(\widetilde{C}_\bullet\otimes_{\mathbb{F}_2[\pi_1]} A_\rho, \widetilde{\partial}_\bullet \otimes \mathbb{1})$ by defining the boundary maps as in the proof of Proposition \ref{proposition universal lift of code as module}, i.e.  $\tilde \partial_2 z=(\tilde d_1 \tilde d^0 z )|_Q $ and $\tilde \partial_1 q=(\tilde d_1 \tilde d^0 q) |_X $, which are also $\mathbb{F}_2[\pi_1]$-module homomorphism. Therefore, tensoring $\tilde \partial_i$ it with the identity map on $A_\rho$ yields a unique $\mathbb{F}_2[\Gamma]$-module homomorphism.\qedhere
\end{proof}
The cochain groups and maps can be defined by considering the dual spaces $\operatorname{Hom}_{\mathbb{F}_2[\pi_1]}(\widetilde{C}_i, A_\rho)$, the set of module homomorphisms from $\widetilde{C}$ to $A$.\par

\begin{remark}
Because the cellular chain complex $C_{H}$ is a right $\mathbb{F}_2[\Gamma]$-module, there is also a trivial balanced product relation, $C_{H}\otimes_{\mathbb{F}_2[\Gamma]} \mathbb{F}_2= C $, corresponding to the space of coinvariants under the action of $\Gamma$, see \cite{Breuckmann2021} Section IV.C.
\end{remark}

\section{Application to product codes}\label{section Applications}

In this section, we analyze our construction for specific instances of input quantum CSS codes: hypergraph product codes. Then we present how the lift can be used as a first step to generate balanced product of quantum CSS codes.

\subsection{Classification of lifts of HPC}\label{section classification of lifts of HPC}

Let $C$ be an HPC constructed from two classical codes $C_1$ and $C_2$. These codes can be represented by their Tanner graphs $\mathcal{T}_1$ and $\mathcal{T}_2$, which have free groups for fundamental group. Our first objective is to classify and analyze the codes that can obtained by lifts of HPC. The two main results of this section are Proposition \ref{theorem classification product} and \ref{Proposition Lift equivalence PK}. Proofs which do not appear in the present section are gathered in Appendix \ref{section proof of lift of HPC}. Here, the lower indices in $C_i$ do not point to the degree in the chain complex, but to a linear code labeled by $i$. \par
Given a Cartesian product of two groups $G=G_1\times G_2$, Goursat quintuples are sets of the form $\{G_{11}, G_{12}, G_{21}, G_{22}, \psi\}$, where $G_{i2}\trianglelefteq G_{i1}\leq G_i$ for $i=1,2$, and $\psi:G_{11}/G_{12}\rightarrow G_{21}/G_{22}$ is an isomorphism. They serve to classify subgroups of $G$, as described later.

\begin{proposition}\label{theorem classification product}
 Let $C=C_1\otimes C_2^*$ be an HPC with $C_1$ and $C_2$ linear codes represented by their Tanner graphs $\mathcal{T}_1$ and $\mathcal{T}_2$. Lifts of $C$ can be classified by the set of Goursat quintuples $\{G_{11}, G_{12}, G_{21}, G_{22}, \psi\}$ corresponding to finite index subgroups of $\pi_1(\mathcal{T}_1)\times \pi_1(\mathcal{T}_2)$, where $G_{12}\trianglelefteq G_{11}\leq \pi_1(\mathcal{T}_1)$ , $G_{22}\trianglelefteq G_{21}\leq \pi_1(\mathcal{T}_2)$, and $\psi:G_{11}/G_{12}\rightarrow G_{21}/G_{22}$ is an isomorphism.
\end{proposition}
The consequence of this proposition is the following. Given an HPC $C$ as above, with $\mathcal{K}:=\mathcal{K}(C)$, and a covering map $p:\mathcal{K}_H\to \mathcal K$, associated to a subgroup $H\leq \pi_1(\mathcal{K})$, suppose $H$ cannot be decomposed as a Cartesian product of subgroups of the factors $\pi_1(\mathcal{T}_1) $ and $ \pi_1(\mathcal{T}_2)$. Then it is clear that $\mathcal{K}_H$ cannot either be decomposed as a Cartesian product of two spaces, since $\pi_1(\mathcal T_i)$ is a free group. We can associate a certain lifted code to each of these non-product subgroups. Although this result was essentially obtained in \cite{Panteleev2021}, we show how to do this lift in the language of covering maps. This section hence establishes that our general notion of lift, Definition \ref{definition lift of quantum CSS code}, agrees with the lifted product construction. All statements of this section related to fundamental group, covering spaces and group of deck transformations can be found in \cite{HatcherTopo}.\par

We first clarify how the Tanner cone complex of $C$ is related to the product of the Tanner graphs $\mathcal{T}_1$ and $\mathcal{T}_2$.

\begin{lemma}\label{lemma fundamental group of product}
    The Tanner cone-complex $\mathcal{K}$ of $C=C_1\otimes C_2^* $ has $\pi_1 (\mathcal{K})\cong\pi_1(\mathcal{T}_1)\times\pi_1(\mathcal{T}_2)$.
\end{lemma}
\begin{proof}
  The Tanner cone-complex of $C$  is obtained from the 2D complex $P:=\mathcal{T}_1\times \mathcal{T}_2$ by dividing every square faces into 2 triangular cells by an edge connecting the $X$ and $Z$ vertices at the corner of the faces.
\end{proof}
The last theorem that we need to state before proving Proposition \ref{theorem classification product} is Goursat Lemma \cite{GOURSAT1889,Anderson2009}.

\begin{lemma}[Goursat]\label{Theorem Goursat}
Let $G_1$ and $G_2$ be groups. There is a one-to-one correspondence between subgroups $\{H\}$ of $G_1\times G_2$ and quintuples  $\{G_{11}, G_{12}, G_{21}, G_{22}, \psi\}$, where each $G_{i2}\trianglelefteq G_{i1}\leq G_i$, and $\psi:G_{11}/G_{12}\rightarrow G_{21}/G_{22}$ is an isomorphism. It is given by  $H=\{(g_1,g_2)\in G_{11}\times G_{21}|\psi (g_1 G_{12})=g_2 G_{22}\}$. \par 
Moreover, $H$ is a normal subgroup iff $G_{i1}$, $G_{i2}$ are normal in $G_i$ and $G_{i1}/G_{i2}\leq \operatorname{Z}(G_i/G_{i2}) $, the center of $G_i/G_{i2}$. 
\end{lemma}
\begin{proof}
For the full proof of this correspondence, see \cite{Anderson2009} Theorem 4, or \cite{Bauer2015} Section 2. Here we only outline the main steps leading to the definition of the bijection.\par
Let $H$ be a subgroup of $G_1\times G_2$.  Define \begin{align}
    G_{11}&=\{g_1\in G_1 | \exists g_2\in G_2, (g_1,g_2)\in H\},\\
    G_{12}&=\{g_1\in G_1 | (g_1,e)\in H\},
\end{align}
and write $G_{21}$ and $G_{22}$ similarly; they satisfy $G_{i2}\trianglelefteq G_{i1} $. It can be proved that  $\psi:G_{11}/G_{12}\rightarrow G_{21}/G_{22}$, defined by $\psi(g_1G_{12})=g_2 G_{22}$ for $(g_1,g_2)\in H$, is an isomorphism. Therefore, we obtain a map sending $H$ to the quintuple $\{G_{11}, G_{12}, G_{21}, G_{22}, \psi\}$. Conversely, given the quintuple $\{G_{11}, G_{12}, G_{21}, G_{22}, \psi\}$, we can check that $H$ is a subgroup of $G_{11}\times  G_{21} $ and therefore of $G_1\times G_2$. The correspondence is established by noticing that the two constructions are inverse of each other.
\end{proof}

We are now ready for the proof of Proposition \ref{theorem classification product}.

\begin{proof}[Proof of Theorem \ref{theorem classification product}]
    Let $\mathcal{K}:=\mathcal{K}(C)$ be the Tanner cone-complex of $C=C_1\otimes C_2^*$, and $G_i=\pi_1(\mathcal{T}_i)$. Then we know by Lemma \ref{lemma fundamental group of product} that $\pi_1 (\mathcal{K})\cong\pi_1(\mathcal{T}_1)\times\pi_1(\mathcal{T}_2)$. By the Galois correspondence, connected covering of $\mathcal{K}$ are classified by subgroups of the fundamental group, which by Goursat Lemma \ref{Theorem Goursat} are classified by Goursat quintuples. Therefore, Goursat quintuples on  $\pi_1(\mathcal{T}_1)\times\pi_1(\mathcal{T}_2)$ yield a complete classification.
\end{proof}

The difficulty is to find families of codes related to these non product subgroups. In the following, we attempt to characterize and construct these subgroups. In our case, Goursat Lemma is applied for $G_1$ and $G_2$ free groups and therefore infinite. The next lemma implies that we can restrict our attention to the set of finite index subgroups and normal subgroups of free groups. \par

\begin{lemma}\label{Lemma index goursat}
Let $H$ be a subgroup of $G_1\times G_2$ corresponding to the quintuple $\{G_{11}, G_{12}, G_{21}, G_{22}, \psi\}$. Then $[G_1\times G_2:H]=[G_1:G_{11}]\cdot[G_2:G_{22}]=[G_1:G_{12}]\cdot[G_2:G_{21}]$.
\end{lemma}
This is a classical result for finite groups, but as far the author knows, there is no standard reference treating the case of infinite groups. We therefore provide a proof in Appendix \ref{section proof of lift of HPC}.\par 

From now on, we do not need the graphs to be bipartite for results to hold, unless it is clear from context that they should be interpreted as Tanner graphs of some linear codes. We explain a method to produce Goursat quintuples corresponding to finite index subgroups of $G_1\times G_2$, with $G_i:=\pi_1(T_i)$ and arbitrary finite graphs $T_i$, for $i=1,2$ . Let $p_{11}: T_{11}\rightarrow T_1$ and $p_{21}: T_{21}\rightarrow T_2$ be two connected covering maps obtained from right derived graphs, generated by permutation voltage or (finite) group voltage. Therefore, $p_1$ and $p_2$ do not need to be regular, but they are finite $d_1$ and $d_2$ covers. Let $\Gamma$ be a finite group of order $d$ and $\nu_i:E^*(T_{i1})\rightarrow \Gamma $ for $i=1,2$,  be two voltage assignments. We note $T_{i2}:=D(T_{i1},\nu_i)$ the right $\Gamma$-lift of $T_{i1}$. The relations between these objects are summarized in the following diagram,
\begin{equation}\label{equation diagram covering}
\begin{tikzcd}
  &   T_{12}\arrow[r,hook]{} \arrow[d, "p_{12}"]
  & T_{12}\times T_{22}
  \arrow[d,"\eta"] 
  &T_{22} \arrow[l,hook'] \arrow[d,"p_{22}"]
\\
  & T_{11} \arrow[d,"p_{11}"]
  & (T_1\times T_2)_H \arrow[d,"p_H"]
  & T_{21} \arrow[d,"p_{21}"]
  \\
  &   T_{1}\arrow[r,hook]{}
  & T_{1}\times T_{2}
  & T_{2} \arrow[l,hook'].
\end{tikzcd}
\end{equation}
The next proposition indicates that Goursat quintuples of finite index subgroups can be obtained from the constituents of Equation \eqref{equation diagram covering}.
\begin{lemma}\label{proposition Goursat Lift}
Let $p_{i1}: T_{i1}\rightarrow T_i$ and $  p_{i2}: T_{i2}\rightarrow T_{i1}$, for $i=1,2$, be degree $d_i$ connected cover constructed as in Equation \eqref{equation diagram covering}.  
\begin{enumerate}
 \item Then, the quintuple $\{G_{11}, G_{12}, G_{21}, G_{22}, \psi\}$, with $G_{i1}=p_{i1*}\pi_1(T_{i1})$,  $G_{i2}=p_{i1*}p_{i2*}\pi_1(T_{i2})$, and $\psi:G_{11}/G_{12}\rightarrow G_{21}/G_{22}$ a choice of isomorphism, is a Goursat quintuple associated to a finite index subgroup $H$ of $\pi_1(T_1)\times \pi_1(T_2)$ of index $d_1.|\Gamma|.d_2$.
 \item Moreover, every Goursat quintuple can be obtained in this way.
 \item If $H$ is normal, then $p_{i1}$ and $p_{i1}\circ p_{i2}$ are normal covers and $p_{i2}$ is an Abelian cover, for $i=1,2$.
\end{enumerate}
\end{lemma}
While this statement appears direct, we still detail a proof, later in this section, using induced maps on fundamental groups. The aim is to better understand the correspondence between the geometrical object, the Goursat quintuple and lifts of graphs, which are central to prove other statements.\par

Note that it is straightforward to obtain a non-trivial Goursat quintuple with this method, i.e. one which is not a product of subgroups. It suffices to take $\Gamma$ a non-trivial finite group, so that $p_{i1*} p_{i2*}\pi_1( T_{i2})$ is a proper subgroup of $p_{i1*} \pi_1( T_{i1})$.\par

We now make a parallel between our approach and the lifted product construction of Panteleev and Kalachev \cite{Panteleev2021}.

\begin{proposition} \label{Proposition Lift equivalence PK}
  Consider the subgroup $H$ associated to the quintuple $\{G_{1}, G_{2}, G_{21}, G_{22}, \psi\}$. Let $\Xi:=\pi_1/H$ and  let $p_{i1}: T_{i1}\rightarrow T_i$ and $  p_{i2}: T_{i2}\rightarrow \Gamma\setminus T_{i2}=T_{i1}$, for $i=1,2$, be connected cover constructed as above.  
  \begin{enumerate}
      \item The $H$-cover of the product complex $T_1\times T_2$ is normal iff $\Gamma$ is Abelian.
      \item The lifted code $C^\Xi:=(C_1\otimes C_2^*)^\Xi$, with $T_i:=\mathcal{T}(C_i)$ for $i=1,2$, is equal to the product code $\mathcal{C}:=C_1^\Gamma\otimes_\Gamma C_2^{\Gamma*}$ obtained by lifting the classical codes according to $p_{i2}$.
  \end{enumerate}
\end{proposition}
The first claim of Lemma \ref{Proposition Lift equivalence PK} is in accordance with a Remark 6 in \cite{Panteleev2021}. In order to prove the results, we need to characterize more precisely the covering associated to a given Goursat quintuple.

\begin{lemma}\label{Proposition Antidiagonal action Goursat}
Let $H$ be the subgroup of $\pi_1:=\pi_1(T_1)\times \pi_1 (T_2)$ associated to the quintuple  $\{G_{11}, G_{12}, G_{21}, G_{22}, \psi\}$ constructed as in Lemma \ref{proposition Goursat Lift}, where we identify directly $\Gamma=G_{11}/G_{12}$. Let $D$ be the set defined as follows $D:=\{(\gamma, \psi(\gamma)), \gamma \in \Gamma\}$. Then $D$ is a multiplicative group and $ \eta : T_{12}\times T_{22}\rightarrow (T_1\times T_2)_H $ is a covering map with $\operatorname{Deck}(\eta)\cong D.$
\end{lemma}

In other terms, Proposition \ref{Proposition Antidiagonal action Goursat} asserts the existence of a free transitive $D$-action on $T_{12}\times T_{22}$ and that $(T_1\times T_2)_H\cong D\setminus T_{12}\times T_{22}$. The action of $D$ is equivalent to a diagonal action of $\Gamma$ on each factor $(\gamma, \psi(\gamma))((c_1,\gamma_1),(c_2,\gamma_2))=((c_1,\gamma \gamma_1 ),(c_2,\psi(\gamma)\gamma_2))$, where the first coordinate represents the $\Gamma$-lift of a cell $c_1\in T_{11}$ and the second coordinate represents the $\Gamma$-lift of a cell $c_2\in T_{12}$. \par
Furthermore, let us consider the following application, $R:\pi_1/G_{12}\times G_{22}\rightarrow \pi_1/H$, which as kernel $\operatorname{ker}(R)=D$. Supposing $H$ is normal in $\pi_1$, $R$ is an epimorphism\footnote{Surjective homomorphism.} and by the first isomorphism theorem,
\[ (\pi_1/G_{12}\times G_{22})/D\cong \pi_1/H.\]
This provides a practical way to assign coordinates for the construction of $(T_1\times T_2)_H$. Let $c_i$ be a cell in $T_i$. Then a cell in $(T_1\times T_2)_H$ can be given by $((c_1,c_2),D(g_1G_{12},g_2 G_{22} ))$, where $g_i\in G_i$.\par 
Our last corollary is a simple application of Proposition \ref{Proposition Lift equivalence PK}. 
\begin{corollary}\label{Corollary linear parameters}
There exists a code $C$ (an HPC), and an indexed family $\{C_n\}_\mathbb{N}$ of LDPC codes with parameters $[[ n,\Theta(n),\Theta(n)]]$ generated by the lift operations (covering maps) on the Tanner cone-complex $\mathcal{K}(C)$.
\end{corollary}
The proof simply consists in translating the constituent of a lifted code in \cite{Panteleev2021} into the present language of covering maps over Tanner graphs.\par

The procedure to obtain the boundary maps in the approach of Section \ref{section Lift of a CSS code: explicit construction} requires fixing a spanning tree in the Tanner cone-complex of $C=C_1\otimes C_2^*$. For completeness, we end this section by showing how to find an explicit spanning tree of a Cartesian product of graph by using the product structure.\par 
We can proceed as follows. Let $S_i=(V_i,E_i)$ be a spanning tree of $T_{i2}$, $i=1,2$ ($E_i$ is a subset of edges in $T_{i2}$ and $V_i$ contains all the vertices). We consider the 1-skeleton of the Cartesian product $(S_1\times S_2)^1=(V=V_1\times V_2, E=E_1\times V_2\cup V_1\times E_2) $ . Then the following graph $T$ is a spanning tree of the Cartesian product $T_{i1}\times T_{i2}$, 
\[T=(S_1\times S_2)^1- (V_1\times E_2- v_1\times E), \]
with $v_1$ an arbitrary choice of vertex in $V_1$. \par

\subsection{Balanced product of quantum CSS codes}\label{section balanced product of quantum CSS codes}

The algebraic complex defined by the balanced product of chain complexes was identified in \cite{Breuckmann2021} to be of major importance in the theory of quantum codes. The lifted product of \cite{Panteleev2020,Panteleev2021} is in fact a particular case of balanced product complex. One difficulty addressed in this work, was the development of a systematic way of generating abstract 3-term chain complexes with a left or a right group action, which is a necessary ingredient to apply the balanced product to quantum codes.\par 

Let $C$ and $D$ be two quantum CSS codes with Tanner cone-complex $\mathcal{K}_1:=\mathcal{K}(C)$ and $\mathcal{K}_2:=\mathcal{K}(D)$. Let $p_{1}: \mathcal{K}_{H_1}\rightarrow \mathcal{K}_1$ and $p_{2}:\mathcal{K}_{H_2}\rightarrow \mathcal{K}_2$ be normal covering maps such that $\pi_1(\mathcal{K}_1)/{H_1}\cong \pi_1(\mathcal{K}_2)/{H_2}$. The associated codes, which satisfy $\mathcal{K}_{H_1}=\mathcal{K}(C^{\Gamma})$ and $\mathcal{K}_{H_2}=\mathcal{K}(D^{\Gamma})$, inherit the left action of the group of deck transformations, which is isomorphic for the two codes, $\operatorname{Deck}(p_{i})=\Gamma$. The action of $\operatorname{Deck}(p_i)$ is properly discontinuous, and the covering are respectively equivalent to $\operatorname{Deck}(p_i)\setminus \mathcal{K}_{H_i}\cong \mathcal{K}_i $. We can also define a right group action of the group of deck transformations by multiplying on the left by the inverse of a group element, or equivalently by the action of the opposite group $\Gamma^{\operatorname{op}}$. \par

The new codes are the chain complexes over $\mathbb{F}_2$, $C ^{\Gamma}=C_2^\Gamma\xrightarrow[]{\partial^C_2} C_1^\Gamma\xrightarrow[]{\partial^C_1} C_0^\Gamma$ and $ D ^{\Gamma}=D_2^\Gamma\xrightarrow[]{\partial^D_2} D_1^\Gamma\xrightarrow[]{\partial^D_1} D_0^\Gamma$. We can construct several new complexes out of these two chain complexes by considering them respectively as right and left $\Gamma$-modules, following Section \ref{section relation to fiber bundle codes}. 
\begin{definition}\label{Definition balanced product of CSS codes}
The balanced product of quantum CSS codes is defined as the total complex of the tensor product double complex, \begin{align}
\mathcal{D}_{\bullet\bullet}:=&  C^{\Gamma}_\bullet \otimes_{\Gamma}  D^{\Gamma}_\bullet.
\end{align}
That is, it is the complex $\mathcal{D}:= C^{\Gamma} \otimes_{\Gamma}  D^{\Gamma}$, with spaces indexed as $\mathcal{D}_{n}= \bigoplus_{i+j=n} C^{\Gamma}_i \otimes_{\Gamma}  D^{\Gamma}_j.$
\end{definition}
We can extract the 3-term complex $\mathcal{D}_{3}\xrightarrow[]{\partial^\mathcal{D}_3}\mathcal{D}_{2}\xrightarrow[]{\partial^\mathcal{D}_2}\mathcal{D}_1$, for which the constituent vector spaces are
\begin{align}
\mathcal{D}_{1}&=(C_0^\Gamma\otimes_\Gamma D_1^\Gamma)\oplus (C^\Gamma_1\otimes_\Gamma D^\Gamma_0), \\
\mathcal{D}_{2}&=(C_2^\Gamma\otimes_\Gamma D_0^\Gamma)\oplus (C_1^\Gamma\otimes_\Gamma D_1^\Gamma)\oplus (C_0^\Gamma\otimes_\Gamma D_2^\Gamma),\\
\mathcal{D}_{3}&=(C_2^\Gamma\otimes_\Gamma D_1^\Gamma)\oplus (C_1^\Gamma\otimes_\Gamma D_2^\Gamma).
\end{align}
The boundary maps are defined on an element $x\otimes y \in C_i^\Gamma\otimes_\Gamma D_j^\Gamma$ by $\partial^\mathcal{D}_{i+j}(x\otimes y)=\partial_i^C x\otimes y+x\otimes \partial_j^D y$ and extended by linearity over chains of $\mathcal{D}_{i+j}$. 

\begin{example}
The easiest variant of this construction involves $C=D$, $H_1=H_2$. More precisely we consider the product of a right $G$-lift of $C$ and its dual  $C^{\Gamma*}$, which  inherits a right linear action of the group of deck transformations, $\mathcal{D}:=  C ^{\Gamma}\otimes_\Gamma C ^{\Gamma*}$.
\end{example}
\begin{remark}
    The balanced product of quantum CSS codes can be obtained from the left-right complex
\[P:=\mathcal{K}_{H_1}\times_\Gamma \mathcal{K}_{H_2}=\mathcal{K}_{H_1}\times \mathcal{K}_{H_2}/\left((x_1,\gamma^{-1}g_1),(x_2,g_2)\right)\sim \left((x_1,g_1) ,(x_2, \gamma g_2), \gamma\in \Gamma )\right).\]
for which we must give an interpretation to the sets of vertices as $X$-checks, qubits and $Z$-checks.
\end{remark}

\section{New constructions and lifts}\label{section New constructions and lifts}

\subsection{General method}

In this section and the next, we introduce non-topological and non-product CSS code constructions which can be lifted into codes with improved parameters. We first explain the general procedure to produce them and, subsequently, we detail explicit examples and compute numerically parameters of some moderate-length lifted codes.\par

For any group $G$, there exists at least one space which has $G$ for fundamental group \cite{HatcherTopo}. When given a presentation $G=\langle S |  R\rangle$, a possible construction of such space, called the\textit{ presentation complex}\footnote{This complex is also related to the Cayley complex of $G$. Let $M_G$ be the complex associated to a presentation $G=\langle S |  R\rangle$. Then $M_G=\operatorname{Cay}^2(G,S)/G$, where $\operatorname{Cay}^2(G,S)$ is the Cayley complex for this presentation. Moreover, $p:\operatorname{Cay}^2(G,S)\to M_G$ is a covering map, and $\operatorname{Cay}^2(G,S)$ is the universal cover of $M_G$. See \cite{HatcherTopo}, Section 1.3. } associated to this presentation, goes as follows.

\begin{lemma}[\cite{HatcherTopo}, Corollary 1.28. ]\label{lemma Presentation complex}
    For every group given with a presentation $G=\langle S |  R\rangle$, there exists a 2-dimensional cell complex $M_G$, the presentation complex, having 1 vertex, $|S|$ edges and $|R|$ 2-cells, with $\pi_1(M_G)=G$. 
\end{lemma}

\begin{proof}
Choose a presentation of $G$, $\langle S|R\rangle$ with $g_i\in S$ and $r_j\in R$. This is the quotient of a free group $F_{|S|}$ on the generators of $S$ by the normal closure of the group generated by $R$. The relations of $R$ are the generators of the kernel of the map $F_{|S|}\to G$. Now construct $M_G$ from a wedge of circles $\vee_i S_i^1$, by attaching one 2-cell, labeled $f_j^2$, along the loops specified by the word $r_j$.
\end{proof}

According to the construction given in the proof, if we suppose $G$ finitely presented, i.e. $G$ admits a presentation $\langle S|R\rangle$ with $S$ and $R$ finite sets, then the associated presentation complex is a finite CW complex. Conversely, if the presentation isn't finite, then the associated complex will have an infinite number of cells. Therefore, we focus on groups which are finitely presented. \par

Homotopy equivalent complexes can also be obtained easily. For example, by taking any other graph with fundamental group $F_{|S|}$, one can add iteratively 2-cells in such a way that relations are added to the fundamental group until it is equal to $\langle S|R\rangle$. Alternatively, it is possible to refine the cellulation of the initial presentation complex of $G$, in order to make it a regular CW complex. In each case, we can regard the obtained cellular chain complex $\mathbb{F}_2F\xrightarrow[]{\partial_2}\mathbb{F}_2E\xrightarrow[]{\partial_1}\mathbb{F}_2V$ as a quantum CSS code, in which the qubits are represented by the edges. \par

As we will see, a particularly interesting cellulation is one with 2-cells for which the boundary is a cycle of 4 edges. We call such a cellulation of $M_G$ a \textit{square cellulation} $M_G^\Box$, and we denote its 1-skeleton $M^{\Box,1}_G$. When an instance $M_G^\Box$ is given, it is sometimes possible to make it finer and define a new \textit{abstract} quantum code over that complex. By abstract, we mean a quantum code in which qubits and checks are not indexed by the dimension of cells\footnote{There is always the FH-construction \cite{Freedman2020CSS_Manifold} which gives as 11-dimensional realization of a code into a manifold cellulation.}. In this work, the codes are said to be of $\operatorname{E}$ or \textit{$\operatorname{V}$-type}, depending on whether qubits are placed on the edges or the vertices. An $\operatorname{E}$-type code simply refers to a topological code, i.e. a code equivalent to the $\mathbb{F}_2$-cellular chain complex of a regular 2-dimensional CW complex $M$. But our intention is to present explicit and non-trivial $\operatorname{V}$-type codes, which are not equivalent to any topological code over a $2$D complex.

\begin{lemma}\label{lemma V-type codes}
There exists an infinite set of group presentations $\{G_i\}$, with $G_i$ admitting a presentation complex $M_{G_i}$ which can be refined into a square cellulation $M_{G_i}^\Box$ and an assignment of checks and qubits to the vertices of $M_{G_i}^\Box$, such that the associated $\operatorname{V}$-type code is not equal to any $\operatorname{E}$-type code over a cellulation of $M_{G_i}$. Moreover, the Tanner cone-complex of this code is homeomorphic to $M_{G_i}^\Box$.
\end{lemma}
More precisely, the Tanner cone-complex of this code is simply obtained from $M_G^\Box$ be subdividing each square 2-cell into 2 triangular cells in the right way, rendering the new complex simplicial.\par

The difficulty to construct such code lies in the fact that the $\operatorname{V}$-type code should not be equal to any $\operatorname{E}$-type code over $M_G$. It is indeed always possible to find a square cellulation of a presentation complex and an assignment of checks and qubits to subsets of vertices, producing a valid CSS code. However, in most cases, the constructed quantum code is equal to an $\operatorname{E}$-type code. The proof of Lemma \ref{lemma V-type codes} consists of presenting specific groups for which the conditions are satisfied. For each of our case study, we first present these groups, labeled $L$, $R$ and $J$, from which we describe the presentation complex. We then define the codes $\operatorname{EL}$, $\operatorname{ER}$, $\operatorname{EJ}$ and $\operatorname{VL}$, $\operatorname{VR}$, $\operatorname{VJ}$, and some specific instances of lift, obtained by covering of their Tanner cone-complex.\par
To compute subgroups and quotient groups, we use the GAP package LINS \cite{GAP4}. To compute an upper bound on the distances $d_X $ and $d_Z$, we use the GAP package $\mathtt{QDistRnd}$ \cite{Pryadko2022}, which operates with a probabilistic method based on information sets.\par
 The criteria that we select to compare the performance of our codes is the quantity $\frac{kd^2}{n}$. For an HPC, it behaves asymptotically as $\Theta (n)$ when $k=\Theta(n)$ and $d=\Theta(\sqrt{n})$, while for a good code, it increases as $\Theta(n^2)$. For the rotated planar code \cite{Bombin_2007}, which is among the most considered codes for practical implementation, it is constant, with $k=1$ and $d^2/n=2$. 

\subsection{Example 1: one-relator group}

\subsubsection{E-type code}
\begin{figure}[t]
  \centering
  \includegraphics[scale=1]{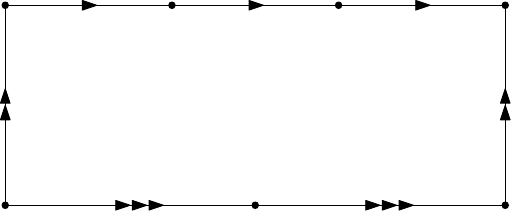}
  \caption{Presentation complex of $M_{2,3}$.}\label{Figure M_2,3}

\end{figure}

In this section, we consider a one-relator group given with presentation 
\[L(a,b):=\langle \: x,y \: |\:  x^a=y^b \:\rangle, \]
for $a,b\geq 2$. The associated presentation complex, labeled $M_{a,b}$, is displayed in Figure \ref{Figure M_2,3} for $a=2$ and $b=3$.  To be precise, instead of starting from the presentation complex (built from a bouquet of 2 circles) we use a complex where 2 circles are linked by an edge. This modification being minor, we still refer to it as the presentation complex.\par
This construction is discussed in Examples 1.24 and 1.35 of \cite{HatcherTopo}. Among its properties, we only mention that, when $a$, $b$ are coprime, $M_{a,b}$ is a deformation retract of the complement of the $(a,b)$-torus knot in $S^3$, and when $a$, $b$ are not coprime, $M_{a,b}$ cannot be embedded in ${\mathbb R}^3$.\par

We specify a regular cellulation of this topological space for each choice of $a,b$ and consider the code defined by its cellular chain complex.

\begin{figure}[t]
  \centering
  \includegraphics[scale=1]{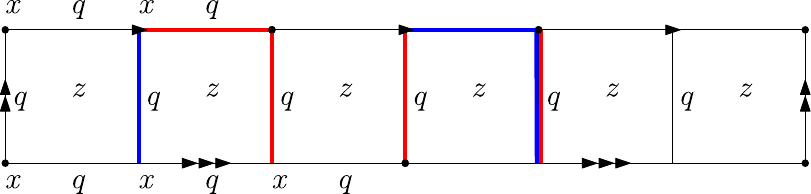}
  \caption{Square cellulation $M_{2,3}^\Box$ corresponding to $\operatorname{EL}(2,3,1)$. Qubits are on the edges, while $X$-checks on the vertices and $Z$-checks on the faces. Non-labeled cells can be inferred from the others. Red edges represent the support of a minimal weight logical $X$-operator, while blue edges represent the support of a minimal weight logical $Z$-operator.}\label{Figure EA(2,3)}

\end{figure}

\begin{definition}\label{Definition EL code}
    Let $2\leq a\leq b$ be integers, $g=\operatorname{gcd}(a,b)$, and $r\geq 1$ be an integer. We define the code $\operatorname{EL}(a,b,r)$ to be the cellular chain complex of the minimal cellulation $M^\Box_{a,b}$ with $r$ rows of $ab/g$ square faces if $g\neq a $, and $2b$ faces if $g=a$ (see Figure \ref{Figure EA(2,3)} for $r=1$). The corresponding chain complex is $
 \operatorname{EL}(a,b,r):=\mathbb{F}_2^{m_Z} \xrightarrow[]{\partial_2=H_Z^T}\mathbb{F}_2^{n} \xrightarrow[]{\partial_1=H_X} \mathbb{F}_2^{m_X}$, such that, if $g\neq a $ :
\begin{itemize}
    \item $n=rab/g+(r-1)ab/g+a/g+b/g$,
    \item $m_X=(r-1)ab/g+a/g+b/g$,
    \item $m_Z=rab/g$.

\end{itemize} 
If $g=a$ :
\begin{itemize}
    \item $n=2+2b/g+2rb+2(r-1)b$,
    \item $m_X=2+2b/g+2(r-1)b$,
    \item $m_Z=2rb$.

    \end{itemize}
In both cases:
\begin{itemize}
    \item     $w_X=b+2$,
    \item $w_Z=4$,
    \item $q_X=2$,
    \item $q_Z=b$.
    \end{itemize}
\end{definition}
In this work, as we only consider the case $r=1$, we simply omit this parameter and note the code $\operatorname{EL}(a,b)$. Its presence in Definition \ref{Definition EL code} is merely intended to understand how cells are counted.

\begin{lemma}\label{lemma dimension EL}The dimension of $\operatorname{EL}(a,b)$ is $k=2$ when $\operatorname{gcd}(a,b)$ is even, and $k=1$ otherwise.
\end{lemma}
\begin{proof}
The highest nontrivial homology group is $H_1(M_{a,b},\mathbb{Z})={\mathbb Z} \oplus {\mathbb Z}_g$, with $g=\operatorname{gcd}(a,b)$. This comes from the Hurewicz theorem, Theorem \ref{theorem Hurewicz}. Indeed, a presentation of the abelianization of $L(a,b)$ is $L^{\operatorname{Ab}}(a,b)=\langle \:x,y \:|\: x^a=y^b,xy=yx\:\rangle$, which can be shown, in a few steps, to be equivalent to $\langle \:s,t \:|\: t^g=1,st=ts\:\rangle$, by specifying $t=x^{a/g}y^{-b/g}$ and $s=x$.\par

The dimension $k$ corresponds to the rank of $H_1(M_{a,b},\mathbb{F}_2)$. We compute this homology, starting from $H_1(M_{a,b},\mathbb{Z})$. The universal coefficient theorem states the existence of the short exact sequence \[0\to H_1(M_{a,b},\mathbb{Z})\otimes_\mathbb{Z} \mathbb{F}_2\to H_1(M_{a,b},\mathbb{F}_2)\to \operatorname{Tor}_1(H_0(M_{a,b},\mathbb{Z}),\mathbb F_2)\to 0.\]The Torsion term is defined as $H_1(E \otimes_\mathbb Z \mathbb{F}_2)$, where $E_\bullet \xrightarrow{\epsilon} H_0(M_{a,b},\mathbb{Z})$ is a free (left) resolution of $H_0(M_{a,b},\mathbb{Z})=\mathbb Z$, the number of connected components of $M_{a,b}$. This 0-th homology group being free, the resolution can be taken to be $0\to E_0\xrightarrow{\epsilon}\mathbb Z\to 0$, with $E_0=\mathbb Z$ and $\epsilon$ is an isomorphism, implying that $E \otimes \mathbb F _2\cong \mathbb F _2$. This chain complex has only trivial maps and therefore no homology, rendering the torsion term trivial, which implies $H_1(M_{a,b},\mathbb{Z})\otimes_\mathbb Z \mathbb F_2 \cong H_1(M_{a,b},\mathbb{F} _2)$. Tensor product being distributive over direct sum of modules, and given that for any two integer $i,j$, $\mathbb Z_i \otimes_\mathbb Z \mathbb Z_j\cong \mathbb Z_{\operatorname{gcd}(i,j)}$, we have $H_1(M_{a,b},\mathbb{F}_2)=\mathbb F_2$ if $g$ odd, and $H_1(M_{a,b},\mathbb{F}_2)=\mathbb F_2^2$ if $g$ even.
\end{proof}

Let $H$ be a finite index subgroup of $L(a,b)$. We define the code $\operatorname{EL}(a,b)_H$ as the $H$-lift of $\operatorname{EL}(a,b)$. Its dimension can be determined by applying Theorem \ref{theorem Hurewicz}, namely we start from $H_1(M_{a,b,H},\mathbb Z)\cong H/[H,H]$, with $M_{a,b,H}$ the $H$-cover of $M_{a,b}$. Then, as in the proof of Lemma \ref{lemma dimension EL}, the universal coefficient theorem yields $H_1(M_{a,b,H},\mathbb{Z})\otimes_\mathbb Z \mathbb F_2 \cong H_1(M_{a,b,H},\mathbb{F} _2)$.\par

There isn't, to our knowledge, a specific mean of calculating the distance of a lifted code using geometric properties of the presentation complex. For that reason, we adopt a numerical approach and generate all lifts of degree 1 to 59 for $3\leq a\leq b\leq 7$. Table \ref{Table EL} gathers a selection of the lifts with the highest values of $kd^2/n$, although not exceeding 1.8. For $a=b=6$, the set of lifts that we report has constant dimension but distance increasing with the length of the code. Most of the other lifts that we haven't mentioned here had highly unbalanced $X$ and $Z$ distances, leading to a poor $kd^2/n$. 

\begin{table}[]
\centering 
\begin{tabular}{l|l|l|l|l|l|l}

$a$& $b$ &$W$& $[\pi_1(\mathcal K ):H]$& $\pi_1(\mathcal K )/H$& $[[n,k,(d_X,d_Z)]]$ & $kd^2/n$ \\ \hline \hline  
  3&   3 &5&    1     &   $\{e\}$&                     $[[10,1,(2,2)]]$&          0.4\\
  &    &&         12&       $A_4$ &                     $[[120,6,(6,6)]$&          1.8\\
 &  && 21& $\mathbb{Z}_7\rtimes\mathbb{Z}_3$& $[[210, 9, (6,6)]]$&1.5\\  
  &    &&         24&       $SL(2,3)$&                     $[[240, 10, (6,6)]]$&          1.5\\ \hline  
  3&   4 &6&         1&       $\{e\}$&                     $[[19,1,(7,3)]]$&          0.5\\ 
  &    &&         36&       $(\mathbb{Z}_3 \times \mathbb{Z}_3)\rtimes\mathbb{Z}_4$&                     $[[684,17,(12,8)]]$&          1.6\\ 
  &    &&         48&       $A_4\rtimes\mathbb{Z}_4$&                     $[[912,22,(12,8)]]$&          1.5\\ \hline  
   6&   6 &8&         1&       $\{e\}$&                     $[[16,2,(2,2)]]$&          0.5\\   
  &    &&         8&       $Q_8$&                     $[[128,2,(8,8)]]$&          1\\
 &  && 12& $\mathbb{Z}_3\rtimes\mathbb{Z}_4$& $[[192,2,(12,10)]]$&1\\ 
  &    &&         32&       $\mathbb{Z}_4.D_8$&                     $[[512,2,(16,16)]]$&          1\\  
  &    &&         40&       $\mathbb{Z}_5\rtimes \mathbb{Z}_8$&                     $[[640,2,(20,18)]]$&          1\\ 
\end{tabular}
\caption{Parameters of lifted codes $\operatorname{VL}(a,b)_H$, for multiple choices of $a,b$. We note $W:=\max (w_X,w_Z,q_X,q_Z)$. A group written $G'.G/G'$ refers to a non split extension of $G/G'$ by $G'$. Numerical results of distance correspond to upper bounds obtained using the GAP package $\mathtt{QDistRnd}$ \cite{Pryadko2022}.}\label{Table EL}
\end{table}

\subsubsection{V-type code}

In this section, we modify the topological construction $\operatorname{EL}(a,b)$ into a new abstract CSS code, i.e. one for which checks and qubits are no longer indexed by the dimension of topological cells.

\begin{definition}\label{Definition VL code}
    Let $2\leq a \leq b$ be integers, $g=\operatorname{gcd}(a,b)$ and $r\geq 1$ be an integer. Consider a square cellulation $M^\Box_{a,b}$ with $2r+1$ rows of 2-cells, and number of vertices $|V|=(2a +2b +4rab)/g$ if $g\neq a$, and $|V|=4+4b/a +8rb$ if $g=a$. We define the code $\operatorname{VL}(a,b,r)$ by assigning both qubits and checks to the vertices of $M^\Box_{a,b}$ in an alternating way, as shown in Figure \ref{Figure VA(2,3)},  together with the identification $\mathcal{T}(\operatorname{VL}(a,b))= M^{\Box,1}_{a,b}$. The corresponding chain complex is $
 \operatorname{VL}(a,b,r):=\mathbb{F}_2^{m_Z} \xrightarrow[]{H_Z^T}\mathbb{F}_2^{n} \xrightarrow[]{H_X} \mathbb{F}_2^{m_X}$,  such that, if $g\neq a$ :
\begin{itemize}
    \item $n=(a+b+2rab)/g$,
    \item $m_X=(b+rab)/g$,
    \item $m_Z=(a+rab)/g$.
\end{itemize}
If $g=a$ :
\begin{itemize}
    \item $n=2+2b/a+4rb$,
    \item $m_X=2b/a+2rb$,
    \item $m_Z=2+2rb$.
\end{itemize}
In both cases :
\begin{itemize}
    \item $w_X=2+a$,
    \item $w_Z=2+b$,
    \item $q_X=b$,
    \item $q_Z=a$.
\end{itemize}
\end{definition}
\begin{figure}[t]
  \centering
  \includegraphics[scale=1]{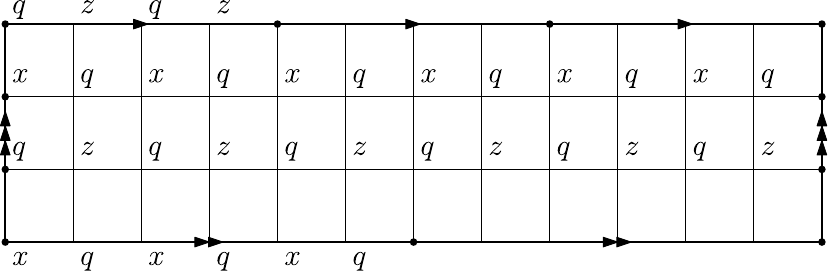}
  \caption{Square cellulation $M^\Box_{2,3}$ corresponding to $\operatorname{VL}(2,3,1)$. Qubits are on the vertices, alternating with $X$-checks and $Z$-checks. Each vertex labeled $q,x,z$ represents a different basis vector. Non labeled vertices can be inferred from labeled ones.}\label{Figure VA(2,3)}
\end{figure}
For $a=2$, respectively $b=2$, the code is topological and $X$-checks, respectively $Z$-checks, can be viewed as faces. Apart from these choices, there are, a priori, no 2D cellular chain complex isomorphic to this code, since the (co)-boundary of a qubit has a weight greater than 2.\par

We emphasize the existence of a code having the following duality property.

\begin{lemma}
    For $a=b$, there is an isomorphism $\operatorname{VL}(a,a,r)\cong\operatorname{VL}^*(a,a,r)$.
\end{lemma}
\begin{proof}
The code has the same number of qubits, $X$ and $Z$-checks. The isomorphism can be obtained by starting from $\operatorname{VL}(a,a,r)$, by translating all the vertex labels ($q,x,$ and $z$) horizontally by one position on the grid, and reflecting them with respect to the central horizontal line (which splits the sets of qubits into two sets of equal size). The obtained code is permutation-equivalent to the dual complex.
\end{proof}
In this work, as we only consider the case $r=1$, we simply omit this parameter and note the code $\operatorname{VL}(a,b)$. It is specifically design to verify the following lemma.
\begin{lemma}
There is an isomorphism $\pi_1(\mathcal{K}(\text{VL}(a,b)))\cong L(a,b)$. 
\end{lemma}
\begin{proof}
    This is direct: the Tanner cone-complex corresponds to a triangulation of $M^\Box_{a,b}$.
\end{proof}
Because this code is not a cellular chain complex, the lift presented in Section \ref{section Lift of a quantum CSS code} is a particularly effective method to build codes of increasing size from a single input code. \par 
As before, for $H$ a (finite index) subgroup of $\pi_1(\mathcal{K})$, we define the code $\operatorname{VL}(a,b)_H$ as the $H$-lift of $\operatorname{VL}(a,b)$. We generate numerically the lifts of degree 1 to 59 for $3\leq a\leq b \leq 7$ and report parameters of input codes and lifted codes in Table \ref{Table VL}. Notice that the dimension of the input codes behaves differently from the one of the $\operatorname{EL}$-codes, which always have $k>0$ as shown in Lemma \ref{lemma dimension EL}. Moreover, although the dimension of the input codes can vanish, this can change in the lifted codes\footnote{Recall that our convention is to set $d=0$ when $k=0$.}. Overall, the parameters of the lifts here are greater than the ones of the $\operatorname{EL}$-codes, which had $kd^2/n$ not exceeding 1.8, and for several input codes we observe a noteworthy increase of this quantity.

\begin{table}[]
\centering 
\begin{tabular}{l|l|l|l|l|l|l}

$a$& $b$&  $W$&$[\pi_1(\mathcal K ):H]$& $\pi_1(\mathcal K ):/H$& $[[n,k,(d_X,d_Z)]]$ & $kd^2/n$ \\ \hline \hline  
  3&   3&     5&1     &   $\{e\}$&                     $[[16,0,0]]$&          0\\
  &   &          &9&       $\mathbb{Z}_3\times\mathbb{Z}_3$&                     $[[144,4,(6,6)]$&          1\\  
  &   &          &27&       $\mathbb{Z}_9\rtimes\mathbb{Z}_3$&                     $[[432,4,(18,18)]]$&          3\\ \hline  
  3&   4&          6&1&       $\{e\}$&                     $[[31,1,(3,7)]]$&          0.3\\ 
  &   &          &48&       $SL(2,3).\mathbb{Z}_2$&                     $[[1488,6,(24,24)]]$&          2.3\\ 
  &   &          &48&       $GL(2,3)$&                     $[[1488,6,(24,24)]]$&          2.3\\ \hline  
   4&   4&          6&1&       $\{e\}$&                     $[[20,2,(2,2)]]$&          0.4\\   
  &   &          &32&       $(\mathbb{Z}_2\times \mathbb{Z}_2). (\mathbb{Z}_4\times\mathbb{Z}_2)$&                     $[[640,6,(16,16)]]$&          2.4\\
 & &  &32& $(\mathbb{Z}_8\rtimes\mathbb{Z}_2)\rtimes \mathbb{Z}_2$& $[[640,6,(16,16)]]$&2.4 \\\hline   
  4&   5&          7&1&       $\{e\}$&                     $[[49,1,(7,5)]]$&          0.5\\  
  &   &          &20&       $\mathbb{Z}_5\rtimes \mathbb{Z}_4$&                     $[[980,7,(20,22)]]$&          2.9\\  
  &   &          &40&       $\mathbb{Z}_5\rtimes \mathbb{Z}_8$&                     $[[1960,7,(40,42)]]$&          5.7\\ \hline
 4& 6&  8&1& $\{e\}$& $[[29,2,(2,3)]]$&0.3\\
 & &  &48& $(\mathbb{Z}_3\times Q_8)\rtimes \mathbb{Z}_2$& $[[1392,8,(20,24)]]$&2.3\\ \hline  
   6&   6&          8&1&       $\{e\}$&                     $[[28,2,(2,2)]]$&          0.3\\  
  &   &          &24&       $\mathbb{Z}_3\times Q_8$&                     $[[672,2,(24,26)]]$&          1.7\\ \hline
 6& 7&  9&1& $\{e\}$& $[[97,1,(10,7)]]$&0.5\\ 
 & &  &42& $\mathbb{Z}_7\rtimes \mathbb{Z}_6$& $[[4074, 11 , (50,50)]]$&6.8\\ 
\end{tabular}
\caption{Parameters of lifted codes $\operatorname{VL}(a,b)_H$, for multiple choices of $a,b$. We note $W:=\max (w_X,w_Z,q_X,q_Z)$. A group written $G'.G/G'$ refers to a non split extension of $G/G'$ by $G'$. Numerical results of distance correspond to upper bounds obtained using the GAP package $\mathtt{QDistRnd}$ \cite{Pryadko2022}.}\label{Table VL}
\end{table}

\subsection{Example 2: group presentation with more relations}\label{Example 2: group presentation with more relations}

\subsubsection{E-type code}\label{section E-type code Example 2}

In this section, we consider a generalization of the previous codes. This is associated to the group given with presentation
\[
R(a,b,c,d):=\langle x_1,\dots,x_{c},y_1,\dots y_{d}\:|\: x_1^a=\dots=x_{c}^a=y_1^b=\dots=y_{d}^b\:\rangle,\]
for integers $a,b\geq 2$ and $c,d\geq 1$. For $c=d=1$, we have $R(a,b,1,1)=L(a,b)$. The presentation complex $M_{a,b,c,d}$, for the code $R(3,2,2,3)$, is depicted in Figure \ref{Figure PCJ(2,3)}.
\begin{figure}[t]
  \centering
  \includegraphics[scale=1]{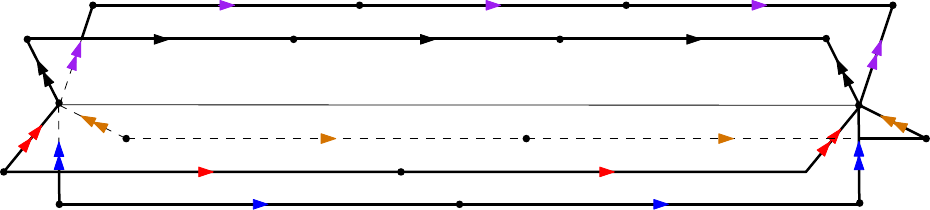}
  \caption{Presentation complex of $R(3,2,2,3)$. Sides decorated with arrows of the same color and type are identified.}\label{Figure PCJ(2,3)}

\end{figure}
Using a square cellulation for this complex, we define the following topological code.
\begin{definition}\label{Definition ER code}
Let $2\leq a \leq b$, $1\leq c,d$ be integers, $g=\operatorname{gcd}(a,b)$. Consider the square cellulation $M^\Box_{a,b,c,d}$, for which each 2-cell of $M_{a,b,c,d}$ is subdivided into a row of $ab/g$ square faces if $g\neq a$, and $2b$ square faces if $g=a$. We define the code $\operatorname{ER}(a,b,c,d)$ to be the cellular chain complex of $M^\Box_{a,b,c,d}$. The corresponding chain complex is 
$\operatorname{ER}(a,b,c,d):=\mathbb{F}_2^{m_Z} \xrightarrow[]{\partial_2=H_Z^T}\mathbb{F}_2^{n} \xrightarrow[]{\partial_1=H_X} \mathbb{F}_2^{m_X}$, such that, if $g\neq a $ :
\begin{itemize}
    \item $n=ab/g+ad/g+bc/g+(c+d)ab/g$,
    \item $m_X=ab/g+ad/g+bc/g$,
    \item $m_Z=(c+d)ab/g$.

\end{itemize} 
If $g=a$ :
\begin{itemize}
    \item $n=2d+2bc/g+2b(c+d)$,
    \item $m_X=2d+2bc/g$,
    \item $m_Z=2b(c+d)$.

    \end{itemize}
In both cases:
\begin{itemize}
    \item     $w_X=\max (b+2,c+d+2)$,
    \item $w_Z=4$,
    \item $q_X=2$,
    \item $q_Z=\max (b, c+d)$.
\end{itemize}
\end{definition}
Note that here, in contrast to $\operatorname{EL}(a,b)$, we only consider the case $r=1$ (number of rows of 2-cells). This is because otherwise the number of qubits would increase linearly in both $r$ and $c+d$.

\begin{remark}
We could consider even more general groups, for which each generator $x_i,y_j$ has its own independent exponent, $a_i,b_j$. Both for simplicity and economy of qubits, we reduce the possible choices by setting $a_i=a$ and $b_j=b$. Indeed, otherwise the number of qubits would scale as $\text{lcm}(a_1,\dots,a_c,b_1,\dots,b_d)$.
\end{remark}

\subsubsection{V-type code: symmetric version} \label{section V-type code: symmetric version}

There are two ways to define an abstract code from the presentation complex of $
R(a,b,c,d)$. In this section, we consider a version that we refer to as symmetric, because under specific choices of integers $a,b,c,d$, the chain complex of the code can be made isomorphic to its dual.

\begin{definition}\label{Definition VR code}
Let $2\leq a \leq b$, $1\leq c,d$ be integers and $g=\operatorname{gcd}(a,b)$. Consider the square cellulation $M^\Box_{a,b,c,d}$, constructed by starting with the one of Definition \ref{Definition ER code}, to which we add a row of square 2-cells in between the faces coming from relations $x_i^a=1$ and the one coming from relations $y_j^b=1$. We define the code $\operatorname{VR}(a,b,c,d)$ by assigning both qubits and checks to the vertices of $M^\Box_{a,b,c,d}$ in an alternating way as shown in Figure \ref{Figure VR(3,3)}, together with the identification $\mathcal{T}(\operatorname{VR}(a,b,c,d))= M^{\Box,1}_{a,b,c,d}$.
The corresponding chain complex is $ \operatorname{VR}(a,b,c,d):=\mathbb{F}_2^{m_Z} \xrightarrow[]{H_Z^T}\mathbb{F}_2^{n} \xrightarrow[]{H_X} \mathbb{F}_2^{m_X}$, such that, if $g\neq a$:
\begin{itemize}
    \item $n=(ad+bc+2ab)/g$,
    \item $m_X=(bc+ab)/g$,
    \item $m_Z=(ad+ab)/g$.
\end{itemize}
If $g=a$:
\begin{itemize}
    \item $n=2d+2bc/a+4b$,
    \item $m_X=2bc/a+2b$,
    \item $m_Z=2d+2b$.
\end{itemize}
In both cases:
\begin{itemize}
    \item $w_X=\max(a+2,d+3)$,
    \item $w_Z=\max(b+2,c+3)$,
    \item $q_X=\max(b,c+1)$,
    \item $q_Z=\max(a,d+1)$.
\end{itemize}
\end{definition}

\begin{figure}[t]
  \centering
  \includegraphics[scale=1]{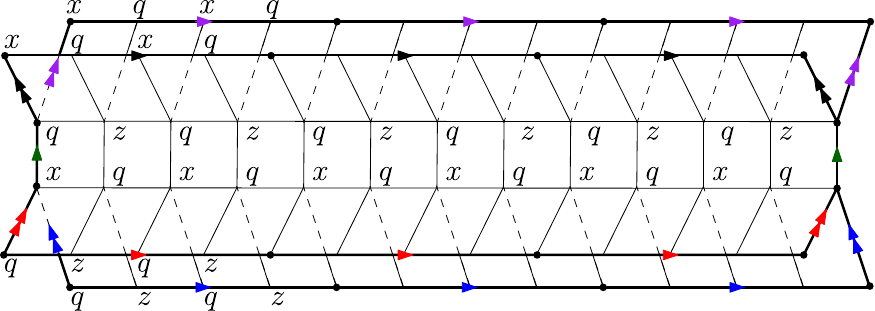}
  \caption{Square cellulation of $M_{2,2,3,3}$ corresponding to $\operatorname{VR}(2,2,3,3)$.  Sides decorated with arrows of the same color and type are identified.}\label{Figure VR(3,3)}

\end{figure}
Among these constructions, we emphasize the existence of the codes $\operatorname{VR}(a,b,c,bc/a)$ when $a|bc$, such as $\operatorname{VR}(a,b,a,b)$, which have the same number of $X$ and $Z$-checks. Another code of interest has the following duality property.

\begin{lemma}
For $a=b$ and $c=d$, we have $\operatorname{VR}(a,a,c,c)\cong \operatorname{VR}^*(a,a,c,c) $.
\end{lemma}
\begin{proof}
    The code has the same number of $X$ and $Z$-checks. The isomorphism can be obtained by starting from  $\operatorname{VR}(a,a,c,c)$, translating all the vertex labels horizontally by one position, and reflecting them with respect to the central horizontal line which splits the sets of qubits into two sets of equal size. This is permutation-equivalent to the dual complex.
\end{proof}

By construction, the Tanner cone-complex of these code satisfies the following lemma.
\begin{lemma}
    There is an isomorphism $\pi_1(\mathcal{K}(\text{VR}(a,b,c,d)))\cong R(a,b,c,d)$. 
\end{lemma}
\begin{proof}
    This is direct: the Tanner cone-complex corresponds to a triangulation of $M_{a,b,c,d}$.
\end{proof}
For  $b>2$ or $c>2$, this code cannot be associated to any cellular chain complex of a $2$D complex because of the value taken by $q_X$ and $q_Z$. Therefore, the lift of Definition \ref{definition lift of quantum CSS code} is a pertinent method to generate new codes. For $H$ a (finite index) subgroup of $\pi_1(K)$, we define the code $\operatorname{VR}(a,b,c,d)_H$ as the $H$-lift of $\operatorname{VR}(a,b,c,d)$. We report parameters of 2 selected lifted codes in Table \ref{Table VR}, due to the enormous amount of subgroups of $R(a,b,c,d)$, that we could not to investigate in its entirety. Observe that, while the relative distance of the $\mathbb Z_3 \times \mathbb Z_ 3$-lift is lower than the one we had for the $\mathbb Z_3 \times \mathbb Z_ 3$-lift of the $\operatorname{VL}$-code,  the quantity $kd^2/2$ is increased by a factor of 2 due to the change of $k$. This suggests that, by adding few qubits in the input code, as controlled by the value of $(c,d)$, the quantity of logical qubits stored in a lifted code can increase drastically.
\begin{table}[]
\centering
\begin{tabular}{l|l|l|l|l|l|l|l|l}
\hline
$a$& $b$ & $c$&$d$ &$W$& $[\pi_1(\mathcal K) :H]$& $\pi_1(\mathcal K)/H$& $[[n,k,(d_X,d_Z)]]$ & $kd^2/n$ \\ \hline\hline
  3&   3 & 2&2 &5&    1     &   $\{e\}$&                     $[[20,2,(2,2)]]$&          0.4\\ 
  &    & & &&         9&       $\mathbb{Z}_3\times\mathbb{Z}_3$&                     $[[180,10,(6,6)]]$&          2\\ \hline 
  4&    4& 2&2 &6&         1&       $\{e\}$&                     $[[24,4,(2,2)]]$&          0.7\\ 
  &    & & &&         12&       $\mathbb{Z}_{12}$&                     $[[288,6,(10,10)]]$&          2.1\\

\end{tabular}
\caption{Parameters of selected lifted codes $\operatorname{VR}(a,b,c,d)_H$, with $W:=\max (w_X,w_Z,q_X,q_Z)$. A group written $G'.G/G'$ refers to a non split extension of $G/G'$ by $G'$. Numerical results of distance correspond to upper bounds obtained using the GAP package $\mathtt{QDistRnd}$ \cite{Pryadko2022}.}\label{Table VR}
\end{table}

\subsubsection{V-type code: asymmetric version}

In this section, we consider a restriction of the group $R(a,b,c,d)$ and a variation of the previous code. For clarity, we define a new group given with the presentation
\[
J(a,b)=\langle x_1,\dots,x_{a}\:|\: x_1^b=\dots=x_{a}^b\:\rangle.\]
In particular, we have $J(a,b)= R(a,a,b-1,1)$. Therefore, the presentation complex for this code is constructed similarly to the previous one. Using a square cellulation for this complex, we define a new $\operatorname{V}$-type code.
\begin{definition}\label{Definition VJ code}
    Let $2\leq a, b$ be integers. Consider the cellulation $M^\Box_{a,a,b-1,1}$, constructed by subdividing each 2-cell of $M_{a,a,b-1,1}$  into $2b$ square faces. We define the code $\operatorname{VJ}(a,b)$ by assigning both qubits and checks to the vertices of $M^\Box_{a,a,b-1,1}$, such that $Z$-checks are assigned to vertices of the central circle, alternating with qubits as shown in Figure \ref{Figure VJ(3,3)}, other qubits and $X$-checks being assigned accordingly, together with the identification $\mathcal{T}(\operatorname{VJ}(a,b))= M^{\Box,1}_{a,a,b-1,1}$. The corresponding chain complex is $ \operatorname{VJ}(a,b):=\mathbb{F}_2^{m_Z} \xrightarrow[]{H_Z^T}\mathbb{F}_2^{n} \xrightarrow[]{H_X} \mathbb{F}_2^{m_X}$, with number of qubits and checks:
\begin{itemize}
    \item $n=2a+2b$,
    \item $m_X=2a$,
    \item $m_Z=2b$.
\end{itemize}
The weights of parity-check matrices are:
\begin{itemize}
    \item $w_X=b+2$,
    \item $w_Z=a+2$,
    \item $q_X=a$,
    \item $q_Z=b$.
\end{itemize}
\end{definition}
\begin{figure}[t]
  \centering
  \includegraphics[scale=1]{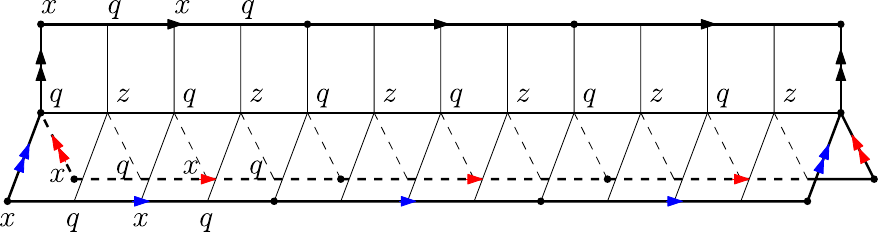}
  \caption{Square cellulation of $M_{3,3}$ corresponding to $\operatorname{VJ}(3,3)$. Sides decorated with arrows of the same color and type are identified.}\label{Figure VJ(3,3)}

\end{figure}
For $a=2$, respectively $b=2$, this code is topological and $Z$-checks, respectively $X$-checks, can be viewed as faces. More precisely, for $a=2$, this code is $\operatorname{EL}(b,b)$. Aside from these choices, there are no 2D cellular chain complex isomorphic to this code. Nevertheless, its Tanner cone-complex still satisfies the following lemma.
\begin{lemma}
   There is an isomorphism $\pi_1(\mathcal{K}(\operatorname{VJ}(a,b)))\cong J(a,b)$.
\end{lemma}
\begin{proof}
   The Tanner cone-complex is the triangulation obtain from $M_{a,a,b-1,1}$ by subdividing each face by a diagonal edge.
\end{proof}

Let $H$ be a (finite index) subgroup of $\pi_1(\mathcal{K})$. We define the code $\operatorname{VJ}(a,b)_H$ as the $H$-lift of $\operatorname{VJ}(a,b)$ and report parameters of a selection of lifted codes in Table \ref{Table VJ}. We only mention a few due to the enormous set of subgroups of $J(a,b)$, that we were not able to investigate in its entirety. Notice that the values of $kd^2/n$ (and $d/n$) for the $\mathbb{Z}_3\times\mathbb{Z}_3$ and  $\mathbb{Z}_9\rtimes\mathbb{Z}_3$-lifts of $\operatorname{VJ}(3,3)$ are slightly increased compared with the corresponding lifts of $\operatorname{VL}(3,3)$.

\begin{table}[]
\centering 
\begin{tabular}{l|l|l|l|l|l|l}
\hline
$a$& $b$ &$W$& $[\pi_1(\mathcal K):H]$& $\pi_1(\mathcal K)/H$& $[[n,k,(d_X,d_Z)]]$ & $kd^2/n$ \\ \hline\hline
  3&   3 &5&    1     &   $\{e\}$&                     $[[12,2,(2,2)]]$&          0.7\\ 
  &    &&         9&       $\mathbb{Z}_3\times\mathbb{Z}_3$&                     $[[108,8,(6,6)]]$&          2.7\\ 
  &    &&         27&       $\mathbb{Z}_9\rtimes\mathbb{Z}_3$&                     $[[324,4,(18,18)]]$&          4\\ \hline 
  4&   3 &6&         1&       $\{e\}$&                     $[[14,3,(2,2)]]$&          0.9\\
 &  && 42& $\mathbb{Z}_2\times ( \mathbb{Z}_7\rtimes \mathbb{Z}_3)$& $[[588,11,(12,12)]]$&2.7\\
\end{tabular}
\caption{Parameters of selected lifted codes $\operatorname{VJ}(a,b)_H$, with $W:=\max (w_X,w_Z,q_X,q_Z)$. A group written $G'.G/G'$ refers to a non split extension of $G/G'$ by $G'$. Numerical results of distance correspond to upper bounds obtained using the GAP package $\mathtt{QDistRnd}$ \cite{Pryadko2022}.}\label{Table VJ}
\end{table}

\section{Discussion and perspectives}

This work started the investigation of lifted codes beyond LPCs. In Section \ref{section Lift of a quantum CSS code}, we have given a definition of the lift based on the idea of covering of the Tanner cone-complex and in Section \ref{section New constructions and lifts}, we introduced a method to construct liftable CSS codes based on the presentation complex of certain groups. Let us now give a few comments and new directions.\par

Firstly, there might exist alternative approaches to lift a CSS code, possibly generalizing Definition \ref{definition lift of quantum CSS code}. There are, undoubtedly, coverings  of the Tanner graph of a given CSS code, which represent valid codes and cannot be obtained as the lifts of the Tanner cone-complex. We believe that such notions might rely more on the use of an algebraic structure rather than geometrical ones.\par

Secondly, we have only introduced examples of non-topological codes associated to a specific instances of the group presentation $R(a,b,c,d)$ of Section \ref{Example 2: group presentation with more relations}. At the time of the writing, we are aware of other group presentations yielding non-topological liftable codes, but these codes,  being of much greater length, are challenging to investigate numerically. It is therefore necessary to understand the behavior of the dimension and distance of a given lifted code analytically. Nevertheless, a less ambitious goal is to devise new examples of moderate length codes, based on certain group presentations, in order to improve on the codes presented here. The parameters exposed in Tables \ref{Table VL}, \ref{Table VR} and \ref{Table VJ} are indeed very far from the ones of the recently introduced group algebra codes \cite{PryadkoGAC2024}, or even Abelian lifted product codes \cite{Panteleev2021degeneratequantum}. Notice, however, that we have only dealt with regular coverings of the Tanner cone-complex, while, as shown in Section \ref{section Applications}, the asymptotically good codes are obtained as non-regular coverings of HPCs.\par

Finally, while it seems that the V-type codes $\operatorname{VR}(a,b,c,d)_H$ and $\operatorname{VJ}(a,b)_H$ are likely to reach interesting parameters at moderate length, we haven't been able to analyze all the lifts, due to the very large amount of them. There is therefore a potential need to generate them by other means. For example, a simple method to produce a subset of them is as follows: given $b$ copies of $\operatorname{EL}(a,a)_H$, we can create a lift of $\operatorname{VJ}(a,b)$ by gluing them along the inverse image (by the covering map) of the central circle in $\mathcal{T}(\operatorname{EL}(a,a))$. We can do a similar construction of certain lifts of $\operatorname{VR}(a,b,c,d)$ , starting with copies of $\operatorname{VL}(a,b)_{H'}$. We leave a more rigorous description and analysis of these codes for future work.

\section*{Acknowledgement}
The author would like to thank Benjamin Audoux and Anthony Leverrier for valuable discussions throughout this work, as well as for carefully reading and commenting the manuscript. The author would also like to thank Allen Hatcher for answering a question on a forum. We acknowledge the Plan France 2030 through the project NISQ2LSQ ANR-22-PETQ-0006.

\bibliographystyle{alpha}
%\bibliography{Biblio_Quantum_Codes}
\newcommand{\etalchar}[1]{$^{#1}$}

\newpage
\appendix

\section{Further restrictions on covering maps}
\begin{figure}[t]
  \centering
  \includegraphics[scale=0.9]{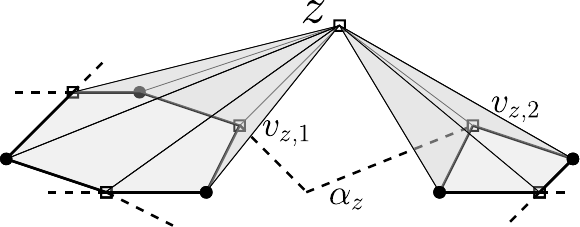}
  \caption{Cone on an induced subgraph $\mathcal{T}_z$ with two disconnected components, linked by a path $\alpha_z$ in $\mathcal{T}(C_X)$.}\label{figure cone-Tz-disconnected}

\end{figure}
In general, for a code $C$ with $\mathcal{T}:=\mathcal{T}(C)$ and $\mathcal{K}:=\mathcal{K}(C)$, some of the induced subgraphs $\mathcal{T}_{z}$ for $z\in Z$, or $\mathcal{T}_{x}$ for $x\in X$, and even the Tanner graphs $\mathcal{T}(C_X)$, $\mathcal{T}(C_Z)$ can be disconnected. This happens, for instance, in Shor's $9$-qubit code \cite{Shor1995}. In this case, when we create the Tanner cone-complex following Proposition \ref{Proposition Tanner cone-complex}, taking the cone on a disconnected subgraph produces new generators in the fundamental group of the complex. 
\begin{lemma}
    Suppose $\mathcal{T}_z$ has $p$ disjoint components, and $\pi_1(\mathcal{T}_z)\cong G$. There exists a subgroup $H\trianglelefteq G$, such that \[\pi_1(\mathcal{T}\cup \operatorname{C}\mathcal{T}_z)\cong G/H * F_{p-1}\]
\end{lemma}
This has no impact on the definition of the lift given in Definition \ref{definition lift of quantum CSS code}, but when doing numerical studies, it can be desirable to produce coverings of $\mathcal{K}$ according to subgroups of $\pi_1(\mathcal{K})$ not involving these generators. To this end, we can trivialize the corresponding generators of the fundamental group by adding new faces in $\mathcal{K}$. In the following procedure, these faces have usually more than 3 edges in their boundaries. Let $z$ be such that $\mathcal{T}_z$ is disconnected, as illustrated in Figure \ref{figure cone-Tz-disconnected}, then
\begin{enumerate}
    \item   Choose a basepoint for each connected component of $\mathcal{T}_{z}$.
    \item For each connected component of $\mathcal{T}(C_X)$, fix an embedded simply-connected path $\alpha_z$ in $\mathcal{T}(C_X)$, going from subgraphs to subgraphs, through each basepoint. For a specific component of $\mathcal{T}(C_X)$ containing $r$ connected components of $\mathcal{T}_z$, label the basepoints according to this path $v_{z,1\leq i\leq r}$, by specifying the starting and end points.
    \item For each reduced subpath of $\alpha_z$ between $v_i$ and $v_{i+1}$, add a (non-triangular) face along the loop formed by $[z,v_{z,i}]$, the subpath from $v_{z,i}$ to $v_{z,i+1}$ and $[v_{z,i+1},z]$.
    \item Repeat steps 1. to 3. for all connected component of $\mathcal{T}(C_X)$ and $\mathcal{T}_z$ and for all $z\in Z$ with disconnected induced subgraphs.
\end{enumerate}
We call the result $\mathcal{K}(C)_{\operatorname{Filled}}$. According to Lemma \ref{lemma restriction covering map}, any covering of  $\mathcal{K}(C)_{\operatorname{Filled}}$ defines one for  $\mathcal{K}(C)$, but not the other way around. Then, everything in Section \ref{section Lift of a CSS code: explicit construction} remains true if we replace $\mathcal{K}(C)$ with  $\mathcal{K}(C)_{\operatorname{Filled}}$.

\section{Proofs for Section \ref{section classification of lifts of HPC}  }\label{section proof of lift of HPC}

We now proceed with the proof of each statement of Section \ref{section classification of lifts of HPC}.

\begin{proof}[Proof (of Lemma \ref{Lemma index goursat}):]
See \cite{Crawford1975} Corollary 4, when the groups are finite. From Lemma \ref{Theorem Goursat}, $H$ has its elements from the set $\{(g_1,g_2)\in G_{11}\times G_{21} | \psi(g_1 G_{12})=g_2 G_{22}\}$. The index $[G_1\times G_2:H]$ is the number of $H$-cosets in $G_1\times G_2/H$. \par
Let $(g_1,g_2)H$ be an element in a $H$-coset. We change coset if we multiply by an element not in $H$. According to the definition, fixing $g_1\in G_{11}$, we change $H$-coset every time we change $G_{22}$-coset, and there are $[G_2:G_{22}] $ of them for each $G_{11}$-coset. Therefore, there are  $[G_1:G_{11}]\cdot[G_2:G_{22}]$ $H$-cosets. 
The second equality is obtained by considering the symmetric definition of the set of $H$, since $\psi$ is an isomorphism.\qedhere
\end{proof}

\begin{proof}[Proof (of Lemma \ref{proposition Goursat Lift})]
\begin{enumerate}
    \item The covering space $ T_{i2}$ being a right $\Gamma$-lift of $  T_{i1}$, it has a discontinuous and free left $\Gamma$-action. Therefore, $p_{i2}$ is the finite regular cover $p_{i2}: T_{i2}\rightarrow  \Gamma \setminus T_{i2}$ with $\operatorname{Deck}(p_{i2})\cong\Gamma$. Moreover, this cover is associated to the subgroup $ p_{i2*}\pi_1( T_{i2})$ of $\pi_1( T_{i1})$ and $\operatorname{Deck}(p_{i2})\cong \pi_1( T_{i1})/ p_{i2*}\pi_1( T_{i2})$. We hence have a first isomorphism $i:\pi_1(T_{11})/ p_{12*}\pi_1( T_{12})\rightarrow \pi_1( T_{12})/ p_{22*}\pi_1( T_{22})$.\par
From the theory of covering maps, the cover $p_{i1}$ being finite connected, it is associated to a finite index subgroup of $\pi_1(T_i)$ and there is an isomorphism $p_{i1*}$ from $\pi_1( T_{i1})$ onto this subgroup, identified with $p_{i1*} \pi_{i}( T_{i1})$. Choosing the normal subgroup $p_{i2*}\pi_1( T_{i2})$ of $\pi_1( T_{i1})$, this induces another isomorphism $\rho_{i}: \pi_1( T_{i1})/ p_{i2*}\pi_1( T_{i2})\rightarrow p_{i1*}\pi_1(T_i)/p_{i1*} p_{i2*}\pi_1( T_{i2})$. Composing this isomorphism with $i$, we have an isomorphism
\[\psi: p_{11*}\pi_1(T_1)/p_{11*} p_{12*}\pi_1( T_{12})\rightarrow p_{21*}\pi_1(T_2)/p_{21*} p_{22*}\pi_1( T_{22}), \]
defined as $\psi:=\rho_2\circ i\circ\rho_1^{-1}$.\par
The covering map $p_i$ being also finite connected, the composition $p_{i1} \circ p_{i2}$ is a well-defined degree $d_i.|\Gamma|$ covering map and induces a homomorphism $(p_{i1} \circ p_{i2})_*=p_{i1*} p_{i2*}$. This chain of isomorphisms defines triples, $p_{i1*} p_{i2*}\pi_1( T_{i2})\trianglelefteq p_{i1*} \pi_1( T_{i1}) \leq \pi_1(T_i) $, as in the proposition. All these subgroups $G_{ij}$ being associated to finite covers, they have finite index in $\pi_1(T_i)$ and are free groups of finite rank given by\[ [\pi_1(T_i):G_{ij}]=\frac{\operatorname{rank}(G_{ij})-1}{\operatorname{rank}(\pi_1(T_i))-1}.\]
Therefore $\psi$ is a valid isomorphism which defines a Goursat quintuple as claimed.\par
It is straightforward to determine the index of the subgroup $H$ associated to this quintuple. Indeed $[\pi_1(T_i):G_{i1}]=d_i$ and $[G_{i1}:G_{i2}]=|\Gamma|$. From, Lemma \ref{Lemma index goursat}, the associated subgroup is of index $[\pi_1(T_1)\times \pi_1(T_2):H]=d_1.|\Gamma|.d_2$.

\item Let  $\{G_{11}, G_{12}, G_{21}, G_{22}, \psi\}$ be any Goursat quintuples of the product of free groups, $G_1\times G_2$ and let $T_i$ be a graph with $\pi_1(T_i)=G_i$. We know from \cite{GROSS1977273} that every subgroup of $\pi_1(T_i) $ can be obtained as an element $p_{i1*} \pi_1( T_{i1})$ for some permutation voltage $T_{i1}$ of $T_i$ such that the associated cover is connected. Moreover, we have access to all normal subgroups of $\pi_1(T_{i1})$ by creating normal connected covers with group voltage $T_{i2}$ of $T_{i1}$. Therefore, we can obtain all subgroups of $p_{i1*} \pi_1( T_{i1})$ in the form of $p_{i2*} \pi_1( T_{i2})$, for some $T_{i2}$. It only remains to fix an isomorphism equal to $\psi$.

\item To obtain a normal subgroup $H$, it is also necessary for $p_{i1}$ and $p_{i1}\circ p_{i2}$ to be connected regular covers. This is direct from the correspondence established in the proof of the 1st claim and Lemma \ref{Theorem Goursat}. Therefore, $p_{i1}$ can also be obtained from a finite connected $\Lambda_i$-lift of $ T_i$ for some group $\Lambda_i$. Moreover, since we must have $G_{i1}/G_{i2}\leq \operatorname{Z}(G_i/G_{i2})$, then $\Gamma$ needs to be Abelian. Notice however that it is not a sufficient condition.\qedhere

\end{enumerate}
\end{proof}

\begin{proof}[Proof (of Proposition \ref{Proposition Antidiagonal action Goursat})]
First, it is easily verified that $G_{12}\times G_{22} \trianglelefteq H $ but in general $G_{12}\times G_{22}$ is not a normal subgroup of $\pi_1$ unless $H$ is a normal subgroup. According to Lemma \ref{Lemma covering of covering}, $\eta: T_{12}\times T_{22}\rightarrow (T_1\times T_2)_H$ is a normal cover with group of deck transformations $\operatorname{Deck}(\eta)\cong\pi((T_1\times T_2)_H)/\eta_*\pi_1( T_{12}\times T_{22})$ and therefore $\operatorname{Deck}(\eta)\setminus(T_{12}\times T_{22})\cong (T_1\times T_2)_H$. \par

 It is straightforward to see, from the definition of $H$, that $D=H/G_{12}\times G_{22}$ but we give several steps for completeness. Indeed, let $Q: \pi_1 \rightarrow \pi_1/G_{12}\times G_{22}$, then $Q$ is not a homomorphism, but its restriction to $H$ is a surjective one. $Q$ can be decomposed as $Q_1\times Q_2$ with $Q_i:G_i\rightarrow G_{i}/G_{i2}$, which are also not homeomorphism. For an element $\gamma\in \Gamma$, $(\gamma,\psi(\gamma))\in G_{11}/G_{12} \times G_{21}/G_{22}$, we show that any element in the preimage $(Q_1^{-1}(\gamma), Q_2^{-1}((\psi(\gamma)))$ is in $H$. Indeed, let $a\in Q_1^{-1}(\gamma)$ and $b\in Q_2^{-1}(\psi(\gamma))$ then $(a,b)\in G_{11}\times G_{21}$. Moreover $\psi (a G_{12})=\psi (\gamma)$ and $Q_2^{-1}(\psi(\gamma))G_{22}=\psi (\gamma)$, hence $(a,b)\in H$ and $Q^{-1}(D)\subseteq H$. Moreover, let $(a,b)\in H$, Then, $\psi(a G_{12})=bG_{22}$ and let us call $\gamma:=a G_{12}$. Since $\psi$ is an isomorphism, we have the unique identification $bG_{22}=\psi(\gamma)$. Hence, each element $(a,b)\in H$ is mapped surjectively by $Q$ onto an element $(\gamma,\psi(\gamma))$, and we showed that  $Q^{-1}(D)= H$. Therefore $D=H/G_{12}\times G_{22}$.\par

We proceed to show that $\operatorname{Deck}(\eta)\cong D$. Indeed, the main covering $p_H$ gives the isomorphism $\pi_1((T_1\times T_2)_H)\cong p_{H*}\pi_1((T_1\times T_2)_H) = H$ and a chain of covering $p_H \circ \eta $ yields the isomorphism $\eta_*\pi_1(T_{12}\times T_{22})\cong p_{H*}\eta_*\pi_1(T_{12}\times T_{22})=G_{12}\times G_{22}$, and $\eta_*\pi_1(T_{12}\times T_{22})\trianglelefteq \pi_1((T_1\times T_2)_H)$. We look at the homomorphism $P:=Q|_H\circ p_{H*}$ where $Q|_H$ is the restriction of $Q$ to $H$. The map $p_{H*}$ is an injective homomorphism, but its restriction to its image is an isomorphism onto $H$, and $Q|_H$ is surjective. Therefore, $P:\pi_1((T_1\times T_2)_H\rightarrow H/G_{12}\times G_{22}$ is an epimorphism and it is direct to compute its kernel. Indeed, $\operatorname{ker} Q=G_{12}\times G_{22}$ yields $\operatorname{ker} P= p_{H*}^{-1}(G_{12}\times G_{22})=\eta_*\pi_1(T_{12}\times T_{22})$. Therefore, by the first isomorphism theorem, 
\[ \pi_1((T_1\times T_2)_H)/\eta_*\pi_1(T_{12}\times T_{22})\cong H/G_{12}\times G_{22}. \]
This finishes the proof by reminding that the l.h.s is isomorphic to $\operatorname{Deck}(\eta)$ and the r.h.s is $D$.  \qedhere
\end{proof}

It is now possible to prove the equivalence between the lifted product code and the lift of a quantum code stated in Proposition \ref{Proposition Lift equivalence PK}

\begin{proof}[Proof (of Proposition \ref{Proposition Lift equivalence PK})]
    \begin{enumerate}
        \item In this setting the subgroups satisfy $G_{i1}=G_i$, hence the covering maps  $p_{i1}: T_{i1}\rightarrow T_i$ are trivial and in particular normal.  For the cover of $T_1\times T_2$ to be normal, we require $H\trianglelefteq \pi_1$, which in this case is reduced to $G_{i}/G_{i2}\leq \operatorname{Z}(G_i/G_{i2})$, and therefore $\Gamma \leq \operatorname{Z}(\Gamma)$.
        \item There is an obvious similarity between the geometrical complexes which represent the code $C^\Xi$ and $\mathcal{C}$. The Tanner graph of $C^\Xi$ is the 1-skeleton of the complex $D\setminus T_{12}\times T_{22}$ while the Tanner graph of $\mathcal{C}$ is the 1-skeleton of the complex $T_{12}\times_\Gamma T_{22}$. These graphs are the same because $D\setminus T_{12}\times T_{22}=T_{12}\times T_{22}/\sim$, with $\sim $ the relation \[((c_1,\gamma_1 ),(c_2,\gamma_2)) \sim ((c_1,\gamma \gamma_1 ),(c_2,\psi(\gamma)\gamma_2)) \Leftrightarrow ((c_1,\gamma^{-1}\gamma_1 ),(c_2,\gamma_2)) \sim ((c_1,\gamma_1 ),(c_2,\psi(\gamma)\gamma_2)),  \]and the equivalence on the r.h.s is the one defining $T_{12}\times_\Gamma T_{22}$. The CSS codes represented by these Tanner graphs are equal if the interpretation that we make of subsets of vertices as qubits, $X$-checks and $Z$-checks are the same. Let $A_i$ and $B_i$ denote the sets of respectively check nodes and bit nodes in $T_{i2}$. The condition is satisfied as for both graphs, the set of qubits is identified with the set $V_1:=B_1\times B_2 \cup A_1\times A_2/\sim$, and the set of checks is identified with the set $V_2:=B_1\times A_2\cup A_1\times B_2/\sim$. \qedhere
    \end{enumerate}
\end{proof}

\begin{proof}[Proof (of Corollary \ref{Corollary linear parameters})]
Let $D_{w}^i$, $i=1,2$, be the bipartite graph with 2 vertices, noted $v_L$ and $v_R$, and $w$ edges $\{e_{1\leq i\leq w}\}$ connecting them. This is the graph considered in \cite{Panteleev2021}, in which bit variables are assigned to the edges, and parity-checks to the vertices. Let $\widehat{D}_{w}^i$ denote the Tanner graphs of these codes $C_i$, where each edge in $D_w^i$ has been split in two by a bit vertex $\{v_{1\leq i\leq w}\}$. For each bit vertex we call \textit{left} and \textit{right} the edges with endpoint respectively the vertex $v_L$ and $v_R$. The code $C$ corresponds to $C_1\otimes C_2^*$.\par
Let $h_1\in \mathbb{F}_2^{r_1\times w}$, $h_2\in \mathbb{F}_2^{r_2\times w}$ be two parity-check matrices verifying the condition in the proof of Proposition 1 of \cite{Panteleev2021} and $T_i:=\mathcal{T}(C(D_w^i,h_i))$ be the Tanner graph of the Tanner code \cite{Sipser1996} defined by the graph $D_w^i$ and the local codes $\operatorname{ker}h_1 $, $\operatorname{ker}h_2$. We can infer a covering of $T_i$ from one of $D_w^i$, reproducing the Tanner code $C(\bar{X}_i,h_i)$ over the Cayley graph $\bar X_i$ with Ramanujan properties \cite{Lubotzky1988} used in \cite{Panteleev2021}. For this, we create a covering of $T_i$ from a specific voltage assignment. \par
Notice that in the Tanner code construction, each edge $e_{i,L}=\{v_i,v_{L}\}$ or $e_{i,R}=\{v_i,v_{R}\}$ in $\widehat{D}_{w}^i$ is replaced by a group of left or right edges $E_{i,L}$ or $E_{i,R}$. Let $\nu$ denotes the voltage assignment on $D_w$ giving rise to the derived graph $\bar X$ and $\hat{\nu}$, the one in $T_i$. If $\nu$ assigns element $\gamma$ to $e_i$ running from $v_L$ to $v_R$, then $\hat{\nu}$  assigns the same element $\gamma$ to all edges of $E_{i,L}^*$ running from $v_L$ to $v_i$ and the identity element to edges of $E_{i,R}^*$. With this choice of $\hat{\nu}$, the derived graph $T_{i2}$ of the cover is equal to the Tanner graph of  $C(\bar{X},h)$.\par
We can then use the covering $T_{i1}\times T_{i2}\to D\setminus T_{i1}\times T_{i2}$ as described in Proposition \ref{Proposition Lift equivalence PK} to create the code $\mathcal{C}=C(\bar{X},h_1)\otimes_\Gamma C^*(\bar{X},h_2)$.\qedhere
\end{proof}

\end{document}